\documentclass[a4paper,11pt]{article}
\usepackage[utf8]{inputenc}
\usepackage[english]{babel}
\usepackage{fullpage}

\usepackage{amsfonts,amssymb,amsmath,mathtools}
\usepackage{graphicx}

\usepackage{amsthm}
\usepackage[inline]{enumitem}
\usepackage[vlined,ruled]{algorithm2e}

\usepackage{hyperref}
\usepackage{cleveref}
\usepackage{mathdots, latexsym, amsbsy, mathtools}
\usepackage{thmtools}
\usepackage{thm-restate}
\usepackage[colorinlistoftodos,textsize=tiny]{todonotes}
\usepackage[nolist]{acronym}
\usepackage{comment}
\usepackage{cancel}
\usepackage{authblk}

\newtheorem{theorem}{Theorem}
\newtheorem{lemma}[theorem]{Lemma}
\newtheorem{claim}[theorem]{Claim}
\newenvironment{claimproof}[1]{\par\noindent\underline{Proof:}\space#1}{\hfill $\blacksquare$}
\newtheorem{corollary}[theorem]{Corollary}
\newtheorem{proposition}[theorem]{Proposition}

\newtheorem{remark}[theorem]{Remark}
\theoremstyle{definition}
\newtheorem{example}[theorem]{Example}
\newtheorem{definition}[theorem]{Definition}

\usepackage{authblk}

\renewcommand{\epsilon}{\varepsilon}
\renewcommand{\phi}{\varphi}
\newcommand{\serc}{\odot_s} %series composition
\newcommand{\parc}{\odot_p} %parallel composition
\newcommand{\comp}{\odot} %composition
\newcommand\encircle[1]{\tikz[baseline=-2.5pt,transform shape, scale=0.6] \node (X) [draw, shape=circle, inner sep=0, minimum size = 3.5ex] {\strut #1};}
\newcommand{\bP}{\boldsymbol{P}}
\newcommand{\blam}{\boldsymbol{\lambda}}
\newcommand{\bsigma}{\boldsymbol{\sigma}}
\newcommand{\bQ}{\boldsymbol{Q}}{}
\newcommand{\bR}{\boldsymbol{R}}
\newcommand{\cost}{\theta}
\newcommand{\makespan}{C_{\max}}
\DeclareMathOperator\val{val}
\newcommand{\Path}[1]{\ensuremath{P_{#1}}}

\DeclareMathOperator\OPT{OPT}
\DeclareMathOperator\Prec{Prec}
\DeclareMathOperator\Succ{Succ}
\DeclareMathOperator\new{new}
\newcommand{\ie}{i.e.,\ }
\newcommand{\eg}{e.g.,\ }
\newcommand{\WLOG}{w.l.o.g.\ }
\newcommand{\Wlog}{W.l.o.g.\ }
\newcommand{\wrt}{w.r.t.\ }

\usepackage[edges]{forest}
\usepackage{tikz}
\usetikzlibrary{arrows.meta}
\usepackage{tikz-3dplot}
\usetikzlibrary{calc,arrows,automata,patterns,graphs,shapes,petri,decorations.pathmorphing,decorations.markings,decorations.pathreplacing} 
\tikzset{->-/.style={decoration={markings,mark=at position .5 with {\arrow{>}}},postaction={decorate}}}
\tikzset{vertex/.style={draw,circle,inner sep=0pt,minimum size=18pt},>=latex'}
\tikzset{c/.style={draw,circle,inner sep=0pt,minimum size=15pt},>=latex'}
\tikzset{dot/.style={draw,circle,inner sep=0pt,minimum size=3pt, fill=black},>=latex'}

\usepackage{xcolor}
\definecolor{rwthblue}{RGB}{0,84,159}
\definecolor{rwthlightblue}{RGB}{142,189,229}
\definecolor{rwthmagenta}{RGB}{227,0,102}
\definecolor{rwthgreen}{RGB}{87,171,39}
\definecolor{rwthorange}{RGB}{246,168,0}
\definecolor{rwthyellow}{RGB}{255,237,0}
\definecolor{rwthred}{RGB}{204,7,30}
\definecolor{rwthbordeaux}{RGB}{161,16,53}
\definecolor{rwthviolett}{RGB}{97,33,88}
\definecolor{rwthpurple}{RGB}{122,111,172}
\definecolor{esayellow}{RGB}{253, 199, 19}

\usepackage[vlined,ruled]{algorithm2e}

\newcommand{\algorithmFont}[1]{\small\ttfamily{#1}}

\SetAlCapSkip{1ex}
\SetAlFnt{\algorithmFont}
\SetNlSkip{1em}
\newcommand{\algorithmInit}{
	\DontPrintSemicolon
	\SetNlSty{lineNumberFont}{}{}
	\SetKwProg{Fn}{function}{}{end}
	\SetArgSty{argFont} 
	\SetKwFunction{KwFn}{function} 
	\SetKwSty{keywordFont}
	\SetCommentSty{commentFont}
}
\SetAlFnt{\small}
\SetAlCapFnt{\small}
\SetAlCapNameFnt{\small}
\SetAlCapHSkip{-1ex}
\SetAlgoCaptionSeparator{}

\usepackage{bm}

\title{On the Approximability of Train Routing and the Min-Max Disjoint Paths Problem}

\author[1]{Umang Bhaskar\thanks{UB was supported by the Department of Atomic Energy, Government of India, under Project No. RTI4001.}}
\author[2]{Katharina Eickhoff\thanks{KE was supported by the Deutsche Forschungsgemeinschaft (DFG, German Research Foundation) – 2236/2.}}
\author[2]{Lennart Kauther\thanks{LK was supported by the German Centre for Rail Traffic Research (DZSF) at the Federal Railway Authority under research project ``Methodik der Kapazitätsbewertung des Gesamtsystems und Knotenberechnung''}}
\author[3]{Jannik Matuschke\thanks{JM was supported by Fonds Wetenschappelijk Onderzoek (FWO) under Research Project G072520N ``Optimization and Analytics for Stochastic and Robust Project Scheduling''}}
\author[2]{Britta Peis}
\author[2]{Laura {Vargas Koch}}
\affil[1]{TIFR Mumbai, India}
\affil[2]{RWTH Aachen University, Germany}
\affil[3]{KU Leuven, Belgium}

\affil[ ]{\newline \textit{umang@tifr.res.in, jannik.matuschke@kuleuven.be, vargaskoch@gdm.rwth-aachen.de, \{eickhoff,kauther,peis\}@oms.rwth-aachen.de}}

\date{July 4, 2025}

\begin{document}
\begin{acronym}
    \acro{TMO}{\textsc{Train Makespan Optimization}}
    \acro{TTPO}{\textsc{Train Throughput Optimization}}
    \acro{MinMaxDP}[\textsc{MinMaxDP}]{\textsc{Min-Max Disjoint Paths}}
    \acro{DAG}{directed acyclic graph}
    \acro{SePa}{series-parallel}
    \acro{MLBA}{\textsc{Multi-level Bottleneck Assignment}}
\end{acronym}
\maketitle
\vspace{-5ex}
\begin{abstract}
In train routing, the headway is the minimum distance that must be maintained between successive trains for safety and robustness. We introduce a model for train routing that requires a fixed headway to be maintained between trains, and study the problem of minimizing the makespan, i.e., the arrival time of the last train, in a single-source single-sink network. For this problem, we first show that there exists an optimal solution where trains move in convoys---that is, the optimal paths for any two trains are either the same or are arc-disjoint.
Via this insight, we are able to reduce the approximability of our train routing problem to that of the min-max disjoint paths problem, which asks for a collection of disjoint paths where the maximum length of any path in the collection is as small as possible.

While min-max disjoint paths inherits a strong inapproximability result on directed acyclic graphs from the multi-level bottleneck assignment problem, we show that a natural greedy composition approach yields a logarithmic approximation in the number of disjoint paths for series-parallel graphs.
We also present an alternative analysis of this approach that yields a guarantee depending on how often the decomposition tree of the series-parallel graph alternates between series and parallel compositions on any root-leaf path.
\end{abstract}

\vspace{1ex}

\noindent
\textbf{Acknowledgements.} We thank Daniel Schmand, Khai Van Tran and Andreas Wiese for fruitful discussions, as well as Nils Nießen and his group for providing us with foundational insights into train planning. We also thank the participants of the Dagstuhl Seminar 24281 \cite{dagstuhl}, the Oberwolfach Workshop in Combinatorial Optimization 2024 \cite{oberwolfach}, and the Chile Summer Workshop on Combinatorial Optimization~2025 for many interesting comments and questions. Finally, we thank the three anonymous reviewers of ESA 2025 for their careful reading of our paper and many useful suggestions.
\vspace{5ex}

\section{Introduction}
\label{sec:introduction}

Given a rail network and a set of trains, the problem of train routing and scheduling refers to the problem of determining a route as well as a schedule for the trains to optimise some objective, such as the total travel time or the maximum travel time (also denoted as makespan).
Train routing usually involves a number of real-world constraints, including train acceleration and deceleration, stopping times at stations, bidirectional arcs, presence of other trains, etc. 
Practical approaches to solving these complex problems often include mixed integer programming, heuristics, and simulations~\cite{Borndoerfer2016, Cacchiani2010,GODWIN2007,Haehn2020,MuD11,PascariuMPAJD24}, see~\cite{LusbyLER11,MuD11,CAIMI2017285} for surveys and an overview.
However, theoretical guarantees on solutions are hard to come by, apart from some results that show \textsf{NP}-hardness already for some very restricted cases~\cite{GODWIN2007, KroonRZ97}.

One of the most important constraints in train routing and scheduling is that of \emph{headway}~\cite{Carey99,Khoshniyat16}, which determines the minimum gap between successive trains on a link. The headway thus corresponds to a safety distance between two successive trains and is important for at least two different reasons. The first is safety: Maintaining sufficient headway between trains reduces the chances of a collision between trains. The second reason is robustness: If a train is delayed somewhat, then the headway can help ensure that this delay is not propagated to all the following trains, and the schedule can recover quickly. In fact, the importance of maintaining a headway is not restricted to train routing. In vehicular traffic, such as cars on roads, traffic often naturally forms a \emph{platoon}---a group of vehicles following each other closely. Here, headway is again essential in maintaining safety of vehicles in the platoon, see~\cite{Sadeghhosseini02} for a survey.

In this paper, our objective is understanding the impact of headway on the complexity of train routing and scheduling.\footnote{Although, as noted, headway plays an important role in routing other vehicles as well, we use the nomenclature for trains as this is what initially motivated our work.} While numerous papers study efficient algorithms for both dynamic and static routing problems, research on efficient algorithms that account for headway is scarce. Towards this, we introduce a basic model of train routing with an exogenously defined headway requirement $\Delta$. 
In our model, we are given a directed network $D=(V,A)$, where each arc $a$ has a travel time $\tau_a \in \mathbb{Z}_{\geq 0}$, and $d$ trains that must travel from a source $s$ to a sink $t$.
Headway constraints require a minimum time separation of $\Delta$ between any two trains using the same arc. Our objective is minimising the makespan---the time the last train reaches the sink $t$. We call this problem \ac{TMO} and refer to Section~\ref{sec.model} for a formal description.

Two interesting special cases of \ac{TMO} arise when the headway is either very large ($\Delta = \infty)$ or very small ($\Delta = 1$).
In the former case, no arc can be used by two trains, and hence the paths used by the trains must be pairwise arc-disjoint. The problem for $\Delta = \infty$ thus becomes the static problem of finding $d$ arc-disjoint $s$-$t$-paths while minimizing the length of the longest path. This is known in the literature as the \ac{MinMaxDP} problem~\cite{LiMS90}.
In particular, the temporal dynamics, \ie the scheduling aspect of the problem, become irrelevant.
Conversely, if $\Delta = 1$, the problem is equivalent to finding an integral \emph{quickest flow over time}, a problem widely studied in traffic networks and transportation~\cite{Skutella2009}. Hence, as $\Delta$ reduces, temporal dynamics play a larger role.

The quickest flow over time problem with aggregated arc capacities~\cite{dressler2010fptas} generalizes the quickest flow over time problem. Here, only a limited amount of flow is allowed to traverse an arc in each time interval of length $\Delta$. 
By setting these capacities to one, \ac{TMO} can be seen as an integral special case of the quickest flow over time problem with aggregated arc capacities.

When studying the computational complexity of the problem from a theoretical perspective, one should note that, for large $d$, the size of the input is $\log(d)$ rather than $d$. Hence, a polynomial time algorithm must run in time polynomial in $\log(d)$ (and the other parameters). Note that this is not even sufficient time to describe a path for each train.
However, we show that optimal solutions can be represented efficiently, by establishing the existence of an optimal \emph{convoy} solution to \ac{TMO} that can be described by a polynomial number of different paths and does not require waiting on any intermediate nodes. 
Furthermore, we show that approximating \ac{TMO}, for any value of~$\Delta$, is equivalent to approximating \ac{MinMaxDP}.F
We further develop new approximation algorithms for \ac{MinMaxDP}, which translate to bounds for \ac{TMO} as well. This work thus provides new results on a classical routing problem and introduces a new model to capture an important parameter in traffic routing and scheduling.

\subsection{Model}\label{sec.model}
We are given a directed graph $D=(V,A)$ with travel times $\tau_a\in \mathbb{Z}_{\geq 0}$ for every $a \in A$, a~source $s \in V$ as well as a sink $t\in V$. We say $|A|=m$. Furthermore, we are given a minimum headway time $\Delta \in \mathbb{Z}_{\geq 1}$ (between the heads of the trains) and a number $d$ of identical trains which we want to send from $s$ to $t$. An instance of train routing is thus denoted $(D,\tau, d, \Delta)$. In the following, we use $[k]\coloneqq\{1,\dots,k\}$ for $k \in \mathbb{Z}_{\geq 1}$ and $\delta^-(v)$ to denote the set of all arcs entering a node $v \in V$. 

A \emph{train routing} $(\bP,\blam)$ is a tuple consisting of a collection  $\bP=(\Path{i})_{i \in [d]}$ of $s$-$t$-paths $\Path{i}$, where $\Path{i}$ consists of the sequence of arcs train $i$ traverses on its way from $s$ to $t$, together with a collection $\blam=(\lambda_i)_{i\in [d]}$ of \emph{entry functions} $\lambda_i: \Path{i} \to \mathbb{Z}_{\geq 0}$,
where $\lambda_i(a)$ denotes the time when the nose of train $i$ enters arc $a\in \Path{i}.$ 

A train routing $(\bP, \blam)$ is \emph{feasible} if it complies with the following two constraints:
\begin{itemize}
    \item $\lambda_i(a) + \tau_a \le \lambda_i(a')$ for all $i \in [d]$ and any pair of consecutive arcs $a, a'$ on $\Path{i}$,
    %$a, a' \in \Path{i}$ such that $a$ is the direct predecessor of $a'$ on $\Path{i}$,
    \item $|\lambda_i(a) - \lambda_{i'}(a)| \geq \Delta$ for all $i, i' \in [d]$ with $i \neq i'$ and $a \in \Path{i} \cap \Path{i'}$.
\end{itemize}
The first constraint implies that the entry function is consistent with the travel times along the path of each train. 
Note that this allows trains to wait at any node.
The second constraint ensures that a feasible routing respects the minimum headway requirement, \ie if two trains use the same arc, they can only do so with a temporal distance of at least $\Delta$.

Given a feasible train routing $(\bP,\blam)$, 
we denote by $C_i(\Path{i}, \lambda_i):=\max_{a\in \Path{i}} (\lambda_i(a)+\tau_a)$ the time when
the nose of train $i$ arrives at $t$. The \emph{makespan} $C_{\max}(\bP,\blam) \coloneqq \max_{i\in [d]} C_i(\Path{i}, \lambda_i)$ is the time when the nose of the latest train reaches the sink $t$.
We consider the \acf{TMO} problem, where the
objective is to find a feasible train routing that minimises the arrival of the last train at $t$. That is, we want to solve
\begin{equation*}
\min\,C_{\max}(\bP,\blam)\ \mbox{ s.t. } (\bP,\blam) \mbox{ is a feasible train routing}.
\end{equation*}

\subsection{The min-max disjoint paths problem}

An interesting special case of \ac{TMO} arises when $\Delta$ is sufficiently large compared to the travel times. In that case any (even approximately) optimal train routing sends all trains along disjoint paths---assuming that the network supports $d$ disjoint $s$-$t$-paths.
This leads us to the \acf{MinMaxDP} problem, which can be formalized as follows: Given a directed graph $D=(V,A)$ with two designated vertices $s,t\in V$ and travel times $\tau:A\to \mathbb{Z}_{\geq 0}$, find $k$ pairwise arc-disjoint $s$-$t$-paths $\Path{1}, \dots, \Path{k}$ minimizing the maximum travel time (\ie makespan) $\max_{i \in [k]} \tau\left(\Path{i}\right),$
where we use $\tau(P) \coloneqq \sum_{a\in P} \tau_a$ to shorten notation for each path $P$.
We call $\bP=(P_1,\dots,P_k)$ a \emph{path profile}.

Early work on \ac{MinMaxDP} has focused on the case where $k$ is constant.
Indeed, Li et al.~\cite{LiMS90}---in the first paper that studied \ac{MinMaxDP}---showed that, even when $k=2$, the problem is strongly \textsf{NP}-hard for general digraphs and weakly \textsf{NP}-hard for \acp{DAG}.
On the positive side, they also showed that \ac{MinMaxDP} can be solved in pseudo-polynomial on \acp{DAG} when $k$ is constant.
By standard rounding techniques, this result can be turned into an FPTAS~\cite{FleischerGLZ07} for the same setting.

Special cases of \ac{MinMaxDP} when $k$ is part of the input have been studied under various names.
Most notably, the \acf{MLBA} problem is equivalent to a special case of \ac{MinMaxDP} in a \ac{DAG} consisting of $b$ layers. Dokka and Goerigk~\cite{DOKKA23} recently established a strong inapproximability bound of $\Omega(b)$ for \ac{MLBA}.
In \Cref{app:hardness}, we observe that their construction also implies that \ac{MinMaxDP} (and, by extension, \ac{TMO}) in \acp{DAG} does not admit approximation factors significantly better than logarithmic in $k$, unless every problem in \textsf{NP} can be solved in quasi-polynomial time (violating the exponential time hypothesis). This leads to the following theorem.
\begin{restatable}{theorem}{hardness}
     \label{thm:hardness}
        There is no $\log^{1-\varepsilon}(k)$-approximation for \ac{MinMaxDP} in \acp{DAG} for any $\varepsilon > 0$, unless $\textsf{NP} \subseteq \textsf{QP}$.
\end{restatable}

Moreover, \ac{MinMaxDP} in bundle graphs, \ie digraphs consisting of a sequence of parallel arcs, is equivalent to the so-called \emph{complete} case of \ac{MLBA}, from which it also inherits strong \textsf{NP}-hardness. For this problem, a long line of approximation algorithms starting in the 1980s~\cite{coffman1984permuting,hsu1984approximation} has recently culminated in a PTAS by Das and Wiese~\cite{SchedulingBagConstWiese} (formulated for an equivalent scheduling problem).

\subsection{Contributions and overview}

In the following, we provide an overview of our  main contributions and the techniques used to derive these results.
\medskip

\noindent\textbf{Existence of optimal convoy routings.}
We show that any instance of \ac{TMO} has an optimal solution in which any two trains either travel along the same path (with sufficient temporal distance), or follow disjoint paths (Theorem~\ref{thm:ConvoyConjecture} in \cref{sec:ConvoySolutions}).
We call such solutions \emph{convoy routings}. 
In particular, convoy routings have a compact representation that is polynomial in the size of the input, even when the number of trains is exponential in the size of the network.
A further consequence is that allowing trains to wait at intermediate nodes does not help to improve the optimal solution value.
The proof of this result relies on an uncrossing algorithm that takes an arbitrary train routing and modifies it to ensure that any train that is the first to traverse the final arc of its path is also the first to traverse any other arc it uses.
Conflicts, \ie situations where this property is violated, are resolved iteratively, starting with those closest to the sink, with a potential function argument ensuring eventual termination of the procedure.
\medskip

\noindent\textbf{A flow-based additive approximation for \ac{TMO}.}
We further provide a simple flow-based approximation algorithm for \ac{TMO} that returns a convoy routing whose makespan is bounded by $\OPT + \Delta$ (\Cref{thm:FoT-APX} in \Cref{sec:FlowAlgo}).
This result is achieved by solving an appropriately defined quickest flow over time problem in the underlying network.
Note that in practice, the headway $\Delta$ is typically much smaller than the total transit time of a train from source to sink, and hence the additive error $\Delta$ is small compared to $\OPT$.
\medskip

\noindent\textbf{Reduction to \ac{MinMaxDP}.}
Building on the previous two results, we establish a strong connection between \ac{MinMaxDP} and \ac{TMO} by showing that any $\alpha$-approximation algorithm for \ac{MinMaxDP} yields a $\max\{1+\epsilon, \alpha\}$-approximation algorithm for \ac{TMO} with polynomial runtime in the size of the input and $\frac{1}{\epsilon}$ (\cref{thm:BlackBoxReduction} in \Cref{sec:Reduction}). 
To achieve this result, we devise a gadget that encodes the choices on the number of trains on each path of a convoy routing into the routing decisions in the network.
The size of this gadget is polynomial in the size of the original network and $d$. It thus provides an exact reduction of \ac{TMO} to \ac{MinMaxDP} when the number of trains $d$ is polynomial in the size of the network.
When $d$ is large, on the other hand, there must be an arc traversed by a large number of trains and hence $\OPT$ is much larger than $\Delta$.
In this case, the additive guarantee of $\OPT+\Delta$ of our flow-based algorithm for \ac{TMO} mentioned above translates to a $(1 + \varepsilon)$-approximation.
As we mentioned previously, for sufficiently large values of $\Delta$, \ac{TMO} coincides with \ac{MinMaxDP}. Thus, finding good approximation algorithms for \ac{MinMaxDP} is not only sufficient but also necessary to obtain good approximations for \ac{TMO}.
\medskip

\noindent\textbf{Approximation algorithms for \ac{MinMaxDP} in series-parallel networks.}
In the remainder of this paper, we therefore focus on approximation algorithms for \ac{MinMaxDP}. 
We show that a natural greedy approach achieves an approximation which is logarithmic in $k$ on \ac{SePa} networks (\Cref{lem:DPepsruntime} in \Cref{sec:Greedy}), where $k$ is the number of paths to be constructed.
This constitutes the first non-trivial approximation result for \ac{MinMaxDP} going beyond the aforementioned PTAS for bundle graphs~\cite{SchedulingBagConstWiese}, without assuming that either $k$ or the number of layers of the graph is constant.

We prove the approximation guarantee via induction on a series-parallel decomposition of the graph. The key ingredient is showing that the greedy algorithm maintains a solution fulfilling a strong balance condition on the path lengths. 
Additionally, we provide an alternative analysis of the same approach that yields an upper bound of $\phi(D) + 1$ on the approximation factor, where~$\phi(D)$ denotes the number of changes from series to parallel composition on a root-leaf path in the binary tree describing the \ac{SePa}-graph.
To the best of our knowledge, the parameter~$\phi(D)$, which we demonstrate to be independent of the choice of the decomposition tree, has not been used in the literature before.
We believe that this natural parameter can be of independent interest for designing and analysing algorithms for other problems in \ac{SePa}-graphs.
Our alternative analysis implies the classic greedy $2$-approximation for bundle graphs by Hsu~\cite{hsu1984approximation} as a special case.
By the aforementioned reduction of \ac{TMO} to \ac{MinMaxDP}, the approximation results for \ac{MinMaxDP} on \ac{SePa}-graphs, translate to our train routing problem.

\section{Existence of optimal convoy routings}
\label{sec:ConvoySolutions}
Let $(D, \tau, d, \Delta)$ be an instance of \ac{TMO}. A \emph{convoy routing} consists of a collection of $k\le d$ pairwise arc-disjoint $s$-$t$-paths $\bP=(P_1,\dots,P_k)$ in $D$, together with a vector $\bsigma \in \mathbb{Z}_{\geq 1}^k$ with $\sum_{i=1}^k \sigma_i=d$, whose $i$-th component $\sigma_i$ denotes the number of trains using path $P_i$. Note that any such convoy routing $(\bP, \bsigma)$ can be turned into a (feasible) train routing $(\bP, \blam)$  by sending $\sigma_i$ trains along path $P_i$ in single file with temporal spacing $\Delta$ between any two consecutive trains.
Thus, the makespan of this train routing is $\max_{i \in [k]} \left\{\tau(P_i) + (\sigma_i - 1) \cdot \Delta\right\}$, where $\tau(P_i):=\sum_{a\in P_i} \tau_a.$

\begin{restatable}[Convoy theorem]{theorem}{ConvoySolution}
    \label{thm:ConvoyConjecture}
    Every instance of \ac{TMO} admits an optimal solution that is a convoy routing.
\end{restatable}

Our proof of~\Cref{thm:ConvoyConjecture} is constructive.
Note that, since we allow waiting at nodes, there is always an optimal train routing $(\bP,\blam)$ which is acyclic in the sense that none of the paths contains a cycle.
In order to show that any instance of \ac{TMO} admits an optimal routing which is a convoy solution (\ie to prove Theorem \ref{thm:ConvoyConjecture}), we take an optimal acyclic train routing $(\bP,\blam)$, and modify this train routing step by step without increasing the makespan until the routing corresponds to a convoy routing.

The algorithm, which we from now on call \textbf{uncrossing algorithm}, proceeds in two phases.
First, in the \textit{uncrossing phase}, we reroute trains such that at the end of the phase, the trains that are the first ones to use their respective ingoing arc to the sink $t$, travel along disjoint paths.
We call these trains \emph{leaders}, and the remaining trains \emph{non-leaders}.
Afterwards, in the \textit{assignment phase}, we assign every non-leader train a route based on the arc it uses to enter~$t$. In particular, each such train follows its leader from $s$ to $t$. 

We now outline the algorithm in more detail and briefly sketch its correctness. A detailed description can be found in \Cref{app:ConvoySolution}.
%A detailed description can be found in the full version of this article~\cite{RoutingTrains_ESA-full}.

\subsection{The uncrossing phase}
We first introduce some notation.
For an arc $a$ entering $t$, i.e., $a \in \delta^-(t)$, we define the \emph{leader} $L_a$ on $a$ as the first train using arc $a$.
We denote by $L = \{L_a\}_{a \in \delta^-(t)}$ the set of leaders.
We say that a leader \emph{$x$ has a conflict on arc $a \in A$} with another train $y$ if $y$ uses some arc $a$ before $x$. 
Note that this relation is not symmetric. 
For a leader $x \in L$, we define the \emph{transition arc} $a^{tr}_x$ as the first arc on which a conflict of $x$ with another train exists when traversing $P_x$ from $t$ to $s$, \ie in the reverse direction.
Note that $a^{tr}_x$ might not exist, in which case $x$ is the first train on every arc on its path.

In the uncrossing phase, we pick a leader $x$ for which $a^{tr}_x$ exists.
We reroute all trains preceding $x$ on $a^{tr}_x$ such that, after rerouting, they all use the same subpath as $x$ from $a^{tr}_x$ to~$t$. The rerouted trains traverse every arc of this $a^{tr}_x$-$t$-path with the same temporal distance to train $x$ as they had on arc $a^{tr}_x$. 
Hence, the rerouted trains keep their respective pairwise temporal distances on the entire path from $a^{tr}_x$ to $t$ and no new conflicts among the rerouted trains arise.
\Cref{fig:reroutedtrains} shows an example of such a rerouting step.
Note that this rerouting procedure also cannot create new conflicts with other trains as $x$ was the first train on every arc of $P_x$ from $a^{tr}_x$ to $t$.
Hence, the new routing is feasible and does not increase the makespan as $x$ arrives after any rerouted train and the routing of $x$ remains unchanged.
\begin{figure}[t]
\centering
    \begin{tikzpicture}[scale=0.69, transform shape,
    every node/.style={draw,circle,inner sep=0pt,minimum size=15pt}]
    
        \node (s) at (0,0) {$s$};
        \node (v1) at (2,0) {$v_1$};
        \node (v2) at (4.5,0) {$v_2$};
        \node (v3) at (6.5,0) {$v_3$};
        \node (t) at (9,0) {$t$};
    
        \begin{scope}[every path/.style={->,draw, very thick}]
            %path of $y$
            \draw[bend left=60, rwthorange,dashed] (s.90) to node[draw=none,fill=none,pos=0.5,above,yshift=-1.5ex] {$\textcolor{rwthgreen}{z}$ $\textcolor{black}{\succ}$ $y$} (v1.90);
            \draw[rwthorange,dashed] (v1.20) to (v2.160);
            \draw[bend left=60, rwthorange,dashed] (v2) to node[draw=none,fill=none,pos=0.5,above] {$y$} (t);

            %path of $z$
            \draw[bend left=60, rwthgreen,dotted] (s) to (v1);
            \draw[rwthgreen,dotted] (v1.-20) -- (v2.200);
            \draw[bend right=60, rwthgreen,dotted] (v2) to node[draw=none,fill=none,pos=0.5,above] {$z$} (v3);
            \draw[rwthgreen,dotted] (v3.-20) -- (t.200);
    
            %path of $x$
            \draw[rwthblue] (s) -- (v1) node[draw=none,fill=none,pos=0.5,above] {$x$};
            \draw[rwthblue] (v1) -- (v2) node[draw=none,fill=none,pos=0.5,above,yshift=-3ex] {\textcolor{black}{ $\textcolor{rwthgreen}{z} \succ$} $x$ \textcolor{black}{ $\succ \textcolor{rwthorange}{y}$}} node[draw=none,fill=none,midway,below,yshift=-1ex,black] {$a^{tr}_x$};
            \draw[rwthblue] (v2) -- (v3) node[draw=none,fill=none,midway,above] {$x$};
            \draw[rwthblue] (v3) -- (t) node[draw=none,fill=none,pos=0.5,above,yshift=-1ex] {\textcolor{black}{ $\textcolor{rwthgreen}{z} \succ$} $x$};

        \end{scope}
    \end{tikzpicture}
    \hfill
    \begin{tikzpicture}[scale=0.69, transform shape,
        every node/.style={draw,circle,inner sep=0pt,minimum size=15pt}]

        \node (s) at (0,0) {$s$};
        \node (v1) at (2,0) {$v_1$};
        \node (v2) at (4.5,0) {$v_2$};
        \node (v3) at (6.5,0) {$v_3$};
        \node (t) at (9,0) {$t$};

        \draw[->,bend left=70] (v2) to (t);
    
        \begin{scope}[every path/.style={->, draw, very thick}]
            %path of $y$
            \draw[bend left=60, rwthorange, dashed] (s.90) to node[draw=none,fill=none,pos=0.5,above,yshift=-1.5ex] {$\textcolor{rwthgreen}{z}$ $\textcolor{black}{\succ}$ $y$} (v1.90);
            \draw[rwthorange,dashed] (v1.20) to (v2.160);
            \draw[rwthorange,dashed] (v2.20) to (v3.160);
            \draw[rwthorange,dashed] (v3.20) to (t.160);

            %path of $z$
            \draw[bend left=60, rwthgreen,dotted] (s) to (v1);
            \draw[rwthgreen,dotted] (v1.-20) -- (v2.200);
            \draw[bend right=60, rwthgreen,dotted] (v2) to node[draw=none,fill=none,pos=0.5,above] {$z$} (v3);
            \draw[rwthgreen,dotted] (v3.-20) -- (t.200);
            
            %path of $x$
            \draw[rwthblue] (s) -- (v1) node[draw=none,fill=none,pos=0.5,above] {$x$};
            \draw[rwthblue] (v1) -- (v2) node[draw=none,fill=none,pos=0.5,above,yshift=-3ex] {\textcolor{black}{ $\textcolor{rwthgreen}{z} \succ$} $x$ \textcolor{black}{ $\succ \textcolor{rwthorange}{y}$}};
            \draw[rwthblue] (v2) -- (v3) node[draw=none,fill=none,midway,above,yshift=-1ex] {$x$ \textcolor{black}{ $\succ \textcolor{rwthorange}{y}$}};
            \draw[rwthblue] (v3) -- (t) node[draw=none,fill=none,pos=0.5,above,yshift=-3ex] {\textcolor{black}{ $\textcolor{rwthgreen}{z} \succ$} $x$ \textcolor{black}{ $\succ \textcolor{rwthorange}{y}$}};
        \end{scope}
    \end{tikzpicture}
\caption{An example iteration of the uncrossing phase. Here $x \succ y$ means that $x$ follows $y$. Train $x \in L$ is chosen by the uncrossing algorithm. As train $y$ precedes $x$ on $a^{tr}_x$, it is rerouted to use the same subpath as $x$ from $a^{tr}_x$ to $t$. The routing after this iteration is shown on the right.}
\label{fig:reroutedtrains}
\end{figure}
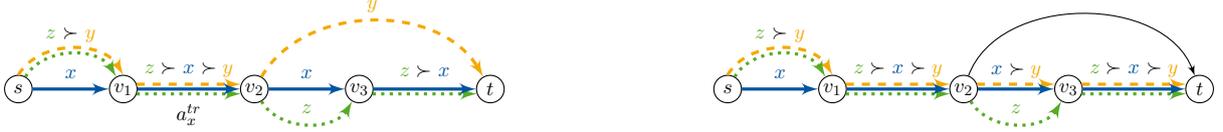

We iterate this procedure with the updated train routing (and the corresponding leaders and transition arcs)
until every leader is the first train on all arcs of its path, which implies that all leaders use disjoint paths. In the proof of \cref{lem:Uncrossing} below, we establish that this termination criterion is reached after a finite number of uncrossing steps.

\subsection{The assignment phase}
After successfully uncrossing the paths of all leaders, we receive an updated routing $(\bP,\blam)$.
To achieve a convoy routing $(\bR,\bsigma)$, define $\bR = (P_x)_{x \in L}$ where $P_x$ is the $s$-$t$-path used by leader $x$ in $\bP$.
We assign all trains that use the same arc as $x$ to enter $t$ in $(\bP,\blam)$ to $P_x$ and define $\sigma_x$ as the resulting number of trains that use $P_x$. 
As each train uses exactly one arc in $\delta^-(t)$, this induces a partition of $[d]$ and thus a feasible convoy routing.

To see why the makespan remains unchanged, note the following:
Fix an arc $a \in \delta^-(t) \cap P_x$ for some $x \in L$.
The first train from our convoy routing $(\bR,\bsigma)$ arrives at $a$ no later than the corresponding leader $x$ in $(\bP,\blam)$.
This is because they use the same path, and the train from $(\bR,\bsigma)$ departs $s$ immediately.
Further, in both routings, the last train entering $a$ must have a minimum temporal distance of $(\sigma_x-1)\cdot \Delta$ to the first train entering $a$. In the convoy routing, this bound is tight and thus the assignment phase does not increase the makespan.

\subsection{Analysis}
By its construction, if the uncrossing algorithm terminates, it returns a convoy routing with a makespan no worse the one of the input routing. Theorem \ref{thm:ConvoyConjecture} thus follows as a consequence of the following lemma.

\begin{lemma}
    \label{lem:Uncrossing}
    The uncrossing algorithm terminates.
\end{lemma}

\begin{proof}[Proof (sketch)]
    We prove this using a potential function.
    To define this function, we extend the definition of $a^{tr}_y$ to arbitrary, \ie\ also non-leader trains, $y$. Informally, $a^{tr}_y$ is the last arc on train $y$'s path after which the set (or order) of trains driving before $y$ changes.

    The potential of a train $y$ is the number of arcs of $P_y$ from $s$ to $a^{tr}_y$.
    We prove the following two statements:
    (i) When resolving a conflict for some leader $x$, we decrease its potential,
    (ii) The potential of all other trains does not increase when resolving a conflict of $x$.
    Together, (i) and (ii) imply that the sum of all potentials decreases in every iteration of the uncrossing phase, implying finite termination of this phase and thus of the entire algorithm.
    The proof of (i) and (ii) involves a careful analysis, which reveals that the transition node of a train is always on its path in the original routing, and a case distinction of possible configurations arising in the uncrossing steps.
    The details of this analysis can be found in Appendix~\ref{app:ConvoySolution}.
    %The details of this analysis can be found in the full version of this article~\cite{RoutingTrains_ESA-full}.
\end{proof}

\section{A flow-based additive approximation for TMO}
\label{sec:FlowAlgo}
As mentioned in \cref{sec:introduction}, \ac{TMO} with $\Delta = 1$ reduces to \textsc{Quickest Flow Over Time} in a network with unit capacities. 
This problem asks for a flow over time of a given value ($d$ units in our case) to be sent from a source to a sink in the shortest time $T$ possible.
For an in-depth discussion of flows over time, we refer to~\cite{Skutella2009}.
By guessing the optimal makespan $T$ (\eg by using parametric search~\cite{saho2017cancel}),
\textsc{Quickest Flow Over Time} can, in turn, be reduced to \textsc{Maximum Flow Over Time}---\ie sending a maximum amount of flow within a given time horizon.
The latter problem can be solved efficiently using a classic result of Ford and Fulkerson~\cite{FordFulkerson1958}: Compute an appropriate static flow, decompose this flow into (disjoint, in the case of unit capacities) paths, and send one unit of flow along each path as long as it still reaches the sink within the time horizon.

A natural extension of this approach to the case $\Delta > 1$ is computing a quickest flow over time sending $\Delta \cdot d$ units in total, and transform it into a train routing by only sending one unit of flow (\ie one train) every $\Delta$ time units along each path. 
Because, conversely, any routing of $d$ trains also can be extended to a flow over time of value $\Delta \cdot d$ when allowing an extra period of $\Delta$ time units beyond its makespan, we obtain the following approximation guarantee for \ac{TMO}; see \cref{app:FlowAlgo} for a complete~proof.
%. Again, a detailed proof can be found in the full version of this article~\cite{RoutingTrains_ESA-full}.

\begin{restatable}{theorem}{FoTapprox}
    \label{thm:FoT-APX}
    There is an algorithm that, given a \ac{TMO} instance, computes in strongly polynomial time a convoy routing with makespan at most $\OPT + \Delta$, where $\OPT$ is the optimal~makespan.
\end{restatable}

\section{Reduction to MinMaxDP}
\label{sec:Reduction}
In this section, we examine the relationship between \ac{TMO} and \ac{MinMaxDP}.
It is easy to see that any $\alpha$-approximation for \ac{TMO} implies an $\alpha$-approximation for \ac{MinMaxDP}, by choosing $\Delta = \alpha \cdot \sum_{a\in A} \tau_a$.
Proving the converse is more involved. We use a bundle graph (see Figure~\ref{fig:prefix_gadget}) consisting of a sequence of parallel arcs as a gadget, attached to the source node for the \ac{MinMaxDP} instance, for our reduction.

\begin{figure}[thbp]
    \centering
    \begin{tikzpicture}[scale=1, every node/.style={draw,circle,inner sep=0pt,minimum size=10pt}]
        \node (s) at (0,0) {$s'$};
        \node (1) at (2,0) {};
        \node (2) at (4,0) {};
        \node (3) at (6,0) {};
        \node[draw=none] (.) at (8,0) {$\dots$};
        \node (t-) at (10,0) {};
        \node (t) at (12,0) {$s$};

        \path[->, draw] (s) edge[in=120, out=60] node[draw=none, pos=0.5,above=-1pt, font=\scriptsize]{$\Delta$} (1);
        \path[->, draw] (s) edge[in=150, out=30] node[draw=none, pos=0.5,above=-1pt, font=\scriptsize]{0} (1);
        \path[->] (s) edge[draw=none, in=200, out=-20] node[draw=none, pos=0.5,above=-1pt, font=\scriptsize]{$\dots$} (1);
        \path[->, draw] (s) edge[in=240, out=-60] node[draw=none, pos=0.5,above=-1pt, font=\scriptsize]{0} (1);

        \path[->, draw] (1) edge[in=120, out=60] node[draw=none, pos=0.5,above=-1pt, font=\scriptsize]{$\Delta$} (2);
        \path[->, draw] (1) edge[in=150, out=30] node[draw=none, pos=0.5,above=-1pt, font=\scriptsize]{0} (2);
        \path[->] (1) edge[draw=none, in=200, out=-20] node[draw=none, pos=0.5,above=-1pt, font=\scriptsize]{$\dots$} (2);
        \path[->, draw] (1) edge[in=240, out=-60] node[draw=none, pos=0.5,above=-1pt, font=\scriptsize]{0} (2);
        
        \path[->, draw] (2) edge[in=120, out=60] node[draw=none, pos=0.5,above=-1pt, font=\scriptsize]{$\Delta$} (3);
        \path[->, draw] (2) edge[in=150, out=30] node[draw=none, pos=0.5,above=-1pt, font=\scriptsize]{0} (3);
        \path[->] (2) edge[draw=none, in=200, out=-20] node[draw=none, pos=0.5,above=-1pt, font=\scriptsize]{$\dots$} (3);
        \path[->, draw] (2) edge[in=240, out=-60] node[draw=none, pos=0.5,above=-1pt, font=\scriptsize]{0} (3);

        \path[->, draw] (t-) edge[in=120, out=60] node[draw=none, pos=0.5,above=-1pt, font=\scriptsize]{$\Delta$} (t);
        \path[->, draw] (t-) edge[in=150, out=30] node[draw=none, pos=0.5,above=-1pt, font=\scriptsize]{0} (t);
        \path[->] (t-) edge[draw=none, in=200, out=-20] node[draw=none, pos=0.5,above=-1pt, font=\scriptsize]{$\dots$} (t);
        \path[->, draw] (t-) edge[in=240, out=-60] node[draw=none, pos=0.5,above=-1pt, font=\scriptsize]{0} (t);
        
    \end{tikzpicture}
    \caption{The gadget $G_k$ for solving \ac{TMO} via \ac{MinMaxDP}. The gadget consists of a sequence of $(d-k)$ bundles. Each bundle contains an arc with travel time $\Delta$ and $k-1$ arcs with travel time 0.
    }
    \label{fig:prefix_gadget}
\end{figure}
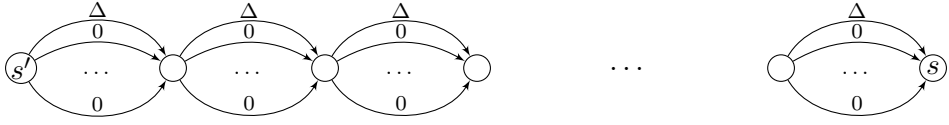

\begin{restatable}{theorem}{BlackBoxReduction}
    \label{thm:BlackBoxReduction}
    Let $\mathcal{A}$ be an $\alpha$-approximation algorithm for $\ac{MinMaxDP}$, and let $\mathcal{R}(m)$ denote its runtime dependent on $m \coloneqq |A|$.
    There exists a $\max\{\alpha, (1+\epsilon)\}$-approximation for \ac{TMO} with runtime in $\mathcal{O}\left(m \cdot \mathcal{R}\left(m^2(2+\frac{1}{\epsilon})\right)\right)$.
\end{restatable}
\begin{proof}[Proof (sketch)]
We defer a formal proof to \Cref{app:Reduction} and restrict to its intuition for now.
%We defer a formal proof to the full version of this article~\cite{RoutingTrains_ESA-full} and restrict ourselves to providing its intuition for now.
First, consider the case $d > m (1+\frac{1}{\epsilon})$, \ie when $d$ is large. Then, $\OPT$ is much greater than $\Delta$ because some arcs are used by more than $\lceil\frac{1}{\varepsilon}\rceil$ trains, and hence the flow-based approximation from Section~\ref{sec:FlowAlgo} yields a $(1+\epsilon)$-approximate solution.

Otherwise, when $d \leq m(1+\frac{1}{\epsilon})$, we utilise the given $\alpha$-approximation $\mathcal{A}$ for \ac{MinMaxDP}.
We guess the number $k \in [m]$ of disjoint paths in an optimal convoy routing. 
For that $k$, we construct an auxiliary digraph $D_k'$ by adding the gadget $G_k$ described in \Cref{fig:prefix_gadget} prior to the source node $s$ of the given digraph $D$.
It is easy to see that the minimum makespan for \ac{MinMaxDP} on $D'_k$ (when $k$ was guessed correctly) and the minimum makespan for the original \ac{TMO}-instance coincide. Hence, $\mathcal{A}$'s approximation guarantee carries over to \ac{TMO}. 
Note that the above case distinction is necessary to bound the size of $D_k'$ polynomially.
\end{proof}

\begin{remark}
    Let $\mathcal{C}$ be a class of digraphs that is closed under adding a bundle graph prior to their source $s$.
    If $\mathcal{A}$ is an $\alpha$-approximation algorithm for \ac{MinMaxDP} on graph class $\mathcal{C}$, then \Cref{thm:BlackBoxReduction} yields an approximation algorithm for \ac{TMO} on graph class $\mathcal{C}$.
\end{remark}

Bundle graphs and \ac{SePa}-graphs are examples of such closed classes.
\Cref{thm:BlackBoxReduction} in conjunction with Theorem~1 from~\cite{SchedulingBagConstWiese} yields the following stronger result for bundle graphs.

\begin{corollary}
There is a PTAS for \ac{TMO} on bundle graphs.
\end{corollary}

\section{Approximation algorithm for \ac{MinMaxDP} in series-parallel graphs}
\label{sec:Greedy}

In this section, we describe and analyse an algorithm for approximating \ac{MinMaxDP} in series-parallel graphs.
Before we state our main result, we first provide a formal definition of series-parallel graphs and the parameter $\phi(D)$ that is part of the approximation guarantee.
\medskip

\noindent
\textbf{Series-parallel graphs.}
A \emph{series-parallel digraph} $D=(V,A)$ (\emph{\ac{SePa}-graph}, for short) has a source $s$ and a sink $t$ and consists of either a single arc or can be obtained by a series or parallel composition of two series-parallel subgraphs $D'$ and $D''$.
The series composition $D = D' \serc D''$ merges the sink of $D'$ with the source of $D''$, such that the source of $D$ is the source of $D'$ and the sink of $D$ is the sink of $D''$.
The parallel composition $D = D' \parc D''$ merges the two sources of $D'$ and $D''$ to a single source of $D$ and merges the sinks of $D'$ and $D''$ to a single sink of $D$. Note that it is possible to have parallel arcs between two nodes.
\medskip

\noindent
\textbf{The parameter $\boldsymbol{\phi(D)}$.} 
Any \ac{SePa}-graph can be described by a \emph{binary decomposition tree}, where the leaves of the tree correspond to the arcs and the internal vertices correspond to either the series (labeled \encircle{$S$}) 
or the parallel (labeled \encircle{$P$}) composition of their children; see  \Cref{fig:sepa-decomposition-tree} for an example.
We can split the tree into arc-, $S$- or $P$-components. Such a component is a maximal connected subset of vertices in the binary decomposition tree with the same label.
The \emph{root} of a component $C$ is the (unique) vertex $v \in C$ that is closest to the root of the binary decomposition tree.
Moreover, we call the vertices that are not in the component $C$ but whose parents are, the children of the component $C$.
We let $\phi(D)$ denote the maximum number of traversed $S$-components on a path from the root to a leaf. Note that, even if the binary decomposition tree of a \ac{SePa}-graph is not unique, the number $\phi(D)$ is (see \Cref{app:phi_unique}).
%(see full version~\cite{RoutingTrains_ESA-full}).
Throughout this section, we assume that we are given a fixed binary decomposition tree of the input graph $D$ (such a tree can be constructed in linear time~\cite{Bodlaender}).

\begin{figure}[tbhp]
    \begin{minipage}{0.56\textwidth}
    \centering
    \begin{tikzpicture}[scale=0.9, transform shape, every node/.style={draw,circle,inner sep=0pt,minimum size=10pt}]
        \node (s) at (0,0) {$s$};
        \node (1) at (2,0) {};
        \node (2) at (3,0) {};
        \node (3) at (4,0) {};
        \node (4) at (6,0) {};
        \node (5) at (7,0) {};
        \node (6) at (8,0) {};

        \path[->, draw] (s) edge[bend left = 35] node[draw=none, rectangle, pos=0.5,above=-1pt, font=\scriptsize]{$a_1$} (1);
        \path[->, draw] (s) edge node[draw=none, rectangle, pos=0.5,above=-1pt, font=\scriptsize]{$a_2$} (1);
        \path[->, draw] (s) edge[bend right = 35] node[draw=none, rectangle, pos=0.5,above=-1pt, font=\scriptsize]{$a_3$} (1);

        \path[->, draw] (1) edge node[draw=none, rectangle, pos=0.5,above=-1pt, font=\scriptsize]{$a_4$} (2);

        \path[->, draw] (2) edge[bend left = 25] node[draw=none, rectangle, pos=0.5,above=-1pt, font=\scriptsize]{$a_5$} (3);
        \path[->, draw] (2) edge[bend right = 25] node[draw=none, rectangle, pos=0.5,above=-1pt, font=\scriptsize]{$a_6$} (3);
        
        \path[->, draw] (1) edge[bend right = 45] node[draw=none, rectangle, pos=0.5,above=-1pt, font=\scriptsize]{$a_7$} (3);

        \path[->, draw] (3) edge[bend left = 35] node[draw=none, rectangle, pos=0.5,above=-1pt, font=\scriptsize]{$a_8$} (4);
        \path[->, draw] (3) edge node[draw=none, rectangle, pos=0.5,above=-1pt, font=\scriptsize]{$a_9$} (4);
        \path[->, draw] (3) edge[bend right = 35] node[draw=none, rectangle, pos=0.5,above=-1pt, font=\scriptsize]{$a_{10}$} (4);
        
        \path[->,draw] (4) edge node[draw=none, pos=0.5,above=-1pt, font=\scriptsize]{$a_{11}$} (5);

        \path[->, draw] (5) edge[bend left = 25] node[draw=none, rectangle, pos=0.5,above=-1pt, font=\scriptsize]{$a_{12}$} (6);
        \path[->, draw] (5) edge[bend right = 25] node[draw=none, rectangle, pos=0.5,above=-1pt, font=\scriptsize]{$a_{13}$} (6);

        \path[->,draw] (s) edge[in=220, out=-40] node[draw=none, pos=0.5,above=-1pt, font=\scriptsize]{$a_{14}$} (6);
    \end{tikzpicture}
    \end{minipage}
    \hspace{-8ex} \hfill
    \begin{minipage}{0.5\textwidth}
    \centering
    \begin{tikzpicture}[scale=0.8, transform shape, every node/.style={draw,circle,inner sep=0pt,minimum size=15pt, font=\scriptsize{#1}}]
        
        \draw[thick, rwthblue, rounded corners=5pt, fill=rwthlightblue!20] (6,5)--(7.5,5.5)--(7.5,6.5)--(6.8,6.5)--(5,5.6)--(3,4.3)--(3,3.5)--(3.8,3.5)--(5,4.5)--(6.8,3.5)--(7.5,3.5)--(7.5,4.5)--cycle;
    
        \node (1) at (1,0) {$a_1$};
        \node (2) at (2,0) {$a_2$};
        \node (p1) at (1.5,1) {$P$};
        \node (3) at (2.5,1) {$a_3$};
        \node[fill=rwthyellow] (p2) at (2,2) {$P$};
        
        \node (4) at (3.5,1) {$a_4$};
        \node (5) at (4,0) {$a_5$};
        \node (6) at (5,0) {$a_6$};
        \node (p3) at (4.5,1) {$P$};
        \node (s1) at (4,2) {$S$};
        \node (7) at (5,2) {$a_7$};
        \node[fill=rwthyellow] (p4) at (4.5,3) {$P$};

        \node (s2) at (3.5,4) {$S$};
        \node (s3) at (5,5) {$S$};
        \node[fill=rwthlightblue] (s4) at (7,6) {$S$};

        \node (s2a) at (7,4) {$S$};
        \node[fill=rwthyellow] (11) at (7.5,3) {$a_{11}$};
        \node[fill=rwthyellow] (p6) at (6.5,3) {$P$};
        \node (p5) at (6,2) {$P$};
        \node (8) at (5.5,1) {$a_8$};
        \node (9) at (6.5,1) {$a_9$};
        \node (10) at (7,2) {$a_{10}$};

        \node (p7) at (8,7) {$P$};
        \node[fill=rwthyellow] (p70) at (8.5,5) {$P$};
        \node (12) at (8,4) {$a_{12}$};
        \node (13) at (9,4) {$a_{13}$};
        \node (14) at (9,6) {$a_{14}$};

        \draw (1) -- (p1);
        \draw (2) -- (p1);
        \draw (p1) -- (p2);
        \draw (3) -- (p2);
        
        \draw (4) -- (s1);
        \draw (p3) -- (s1);
        \draw (5) -- (p3);
        \draw (6) -- (p3);
        \draw (7) -- (p4);
        \draw (s1) -- (p4);
        
        \draw (8) -- (p5);
        \draw (9) -- (p5);
        \draw (p5) -- (p6);
        \draw (10) -- (p6);
        \draw (p6) -- (s2a);
        \draw (11) -- (s2a);
        \draw (s2a) -- (s3);
        
        \draw (p4) -- (s2);
        \draw (p2) -- (s2);
        \draw (s2) -- (s3);
        \draw (12) -- (p70);
        \draw (13) -- (p70);
        \draw (14) -- (p7);
        \draw (s3) -- (s4);
        \draw (p70) -- (s4);

        \draw[rwthorange, ultra thick] (5)--(p3)--(s1)--(p4)--(s2)--(s3)--(s4)--(p7);

        \node[thick, rwthblue, draw=none, rectangle] at (3.7,6) {\normalsize{$S$-components}};
        \draw[thick, rwthblue] (s1) circle (12pt);
        
        \draw[thick, rwthgreen,dashed] (p3) circle (12pt);
        \draw[thick, rwthgreen,dashed] (p7) circle (12pt);
        \draw[thick, rwthgreen,dashed] (p4) circle (12pt);
        \draw[thick, rwthgreen,dashed] (p70) circle (12pt);
        \draw[thick, rwthgreen,dashed, rounded corners=11pt,rotate around={-27:(1.75,1.5)}] (2.2, 2.5) rectangle (1.3, 0.5);
        \draw[thick, rwthgreen,dashed, rounded corners=11pt,rotate around={-27:(5.75,2.5)}] (5.8, 3.8) rectangle (6.7, 1.7);
        \node[thick, rwthgreen,dashed, draw=none, rectangle] at (5.8,7.2) {\normalsize{$P$-components}};
    \end{tikzpicture}
    \end{minipage}
    \caption{An example of a \ac{SePa}-graph on the left. A corresponding series-parallel decomposition tree is shown on the right. 
    Consider the large $S$-component (shaded in light blue), its root is highlighted blue and its children are highlighted yellow.
    Along the orange (bold) path (from the root to a leaf),
    we traverse two $S$-components. Also every other path traverses at most two, so $\phi(D) = 2$.}
    \label{fig:sepa-decomposition-tree}
\end{figure}
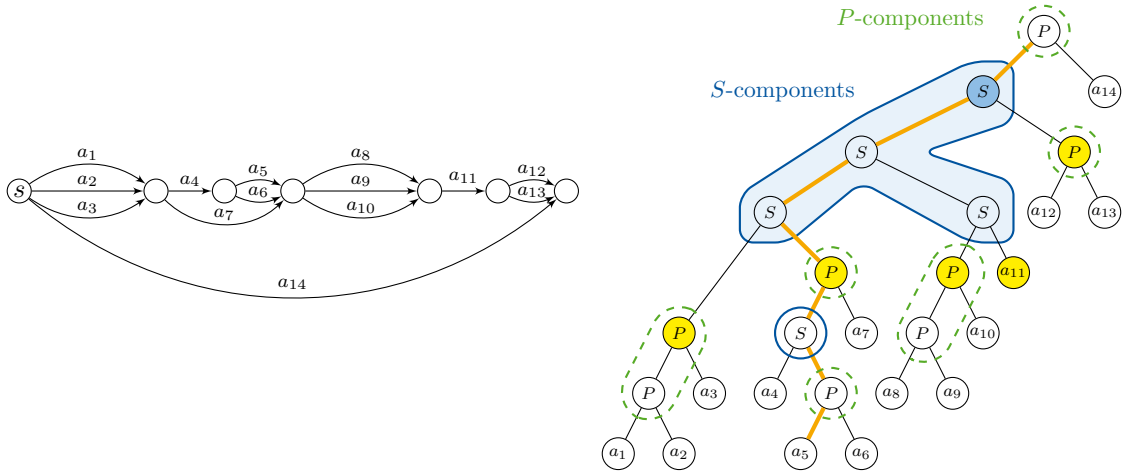

\medskip

\noindent
{\textbf{Main result and overview.}}
We can now formally state our main result for this section. Note that $H_k$ denotes the $k$-th harmonic number, \ie $H_k = \sum_{j=1}^k\frac{1}{j} \leq \log_2(k)+1$.
\begin{restatable}{theorem}{DPepsruntime}
    \label{lem:DPepsruntime}
    For every $\epsilon >0$ there is an algorithm that computes a $(\min {\{H_k, \phi(D)+1\})} \cdot (1+\epsilon)$-approximate path profile for \ac{MinMaxDP} and runs in time polynomial in $m$, $k$, and~$\frac{1}{\varepsilon}$.
\end{restatable}

A central component of our algorithm is a \emph{greedy composition} procedure that inductively combines path profiles in subgraphs based on the binary decomposition tree.
In \Cref{sec:greedy-composition}, we present this greedy composition procedure and several invariants maintained by the procedure that are crucial to our analysis.
In particular, we show that if the invariants are fulfilled by the final solution, they imply the desired approximation ratio.

A remaining challenge is that it is not obvious how to initialise the procedure with solutions that satisfy the invariants at the beginning.
In particular, some of the invariants are defined with respect to certain properties of the optimal solution and are thus not easy to verify.
In \Cref{sec:DP}, we show how to solve this issue by devising two dynamic programs, each optimising for one of the two terms whose minimum defines the approximation guarantee.
Running these two DPs concurrently then yields the guarantee of the theorem (where $\varepsilon$ is an error incurred by rounding so that the DPs run in polynomial time).

\subsection{Greedy composition procedure}
\label{sec:greedy-composition}

We use the following notation.
We write $D' \sqsubseteq D$ for a series-parallel subgraph $D'$ which corresponds to one of the vertices of the decomposition tree of $D$, \ie the decomposition tree of $D'$ is the subtree rooted at that vertex.
Now consider any such subgraph $D' \sqsubseteq D$ with arc set $A'$, source $s'$, and sink $t'$.
If an $s$-$t$-path $P$ in a path profile $\bP$ of $D$ uses an arc in $A'$, the series-parallel structure of $D$ implies that $P \cap A'$ is an $(s',t')$-path in $D'$.
We denote the length of a path $P$ in $D'$ by $\tau(P,D') \coloneqq \sum_{a \in P \cap A'} \tau_a$, and the length of the longest subpath of $\bP$ in $D'$ by $\makespan(\bP,D') \coloneqq \max_{P \in \bP} \tau(P,D')$.
Moreover, we call $\cost(\bP, D')$ the total length of arcs of path profile $\bP$ in component $D'$, that is $\cost(\bP, D')\coloneqq \sum_{P \in \bP} \tau(P,D')$.
\medskip

\noindent
\textbf{Greedy composition.} 
Let $D'$ and $D''$ be two \ac{SePa}-graphs, and let $\bP'$ and $\bP''$ be path profiles of $D'$ and $D''$ consisting of $k'$ and $k''$ arc-disjoint paths, respectively. 
For a series composition $D=D' \serc D''$, if $k'=k''$, the \emph{greedy composition} of $\bP$ and $\bP'$ is the path profile $\bP = \bP' \serc \bP''$ of $D$ that is derived by combining the longest path in $\bP'$ with the shortest path in $\bP''$, the second-longest path in $\bP'$ with the second-shortest path in $\bP''$, and so on. Note that this construction is only defined for $k'=k''.$
For a parallel composition $D=D' \parc D''$, the greedy composition is given by $\bP = \bP' \parc \bP''$ of $D$, where $\bP$ simply consists of the union of the $k'+k''$ paths in $\bP'$ and $\bP''$.
\medskip

We now discuss three properties that are maintained by the greedy composition procedure described above if they are fulfilled by the path profiles that are being combined.
To define these properties, we fix an arbitrary optimal path profile $\bP^*$ in the graph $D$ that defines our \ac{MinMaxDP} instance.
The first property we consider is consistency.
\medskip

\noindent\textbf{Consistency.}
    We call a path-profile $\bP$ \emph{consistent} with $\bP^*$ in $D' \sqsubseteq D$  if $\bP$ and $\bP^*$ use the same number of paths in $D'$ and if
    \begin{equation}
        \label{eq:inv1}
        \cost(\bP, D') \leq \cost(\bP^*, D').
    \end{equation}
Note that consistency is maintained by greedy composition, as the arcs used by the disjoint paths in the composed profile are exactly the union of the paths in the two profiles that are being combined.
We now define the second property, balancedness, and show that together with consistency it implies the first of our two approximation bounds.
\medskip

\noindent\textbf{Balancedness.}
    Let $\bP$ be a path profile, and let $D'\sqsubseteq D$ be traversed by $k'\le k$ paths from $\bP$. Let $p_1\ge \ldots \ge p_{k'}$ be the lengths of those paths in $D'$. Then $\bP$ is \emph{balanced} in $D'$ if
     \begin{equation}
         \label{eq:inv2}
        \frac{1}{i} \left(\sum_{j=1}^i p_j \right) - p_{i+1} \leq \makespan(\bP^*,D') \qquad \text{ for all } i \in [k'-1].
    \end{equation}
We first establish that balancedness---when combined with consistency---is indeed maintained by greedy composition.
While this is straightforward for parallel composition, the proof for series composition requires a more involved analysis that carefully combines \eqref{eq:inv1} and \eqref{eq:inv2}. The details of this analysis are given in \Cref{app:lem:H_k-combination}.

\begin{restatable}{lemma}{HkCombination}
    \label{lem:H_k-combination}
    Let $\comp \in \{\serc, \parc\}$ and let $D^{(1)} \comp D^{(2)}=D'$.
    If $\bQ$ and $\bR$ are consistent and balanced path profiles in $D^{(1)}$ and $D^{(2)}$, respectively, then the greedy profile $\bP = \bQ \comp \bR$ is a consistent and balanced path profile in $D'$.
\end{restatable}

Crucially, the following lemma reveals that if the final solution is balanced and consistent with $\bP^*$, it is a logarithmic approximation for the latter.
The lemma follows from an averaging argument over \eqref{eq:inv2} for $i \in [k-1]$, which can be found in \Cref{app:lem:H_k-approx}.
%the full version of this article~\cite{RoutingTrains_ESA-full}
\begin{restatable}{lemma}{HkApprox}
    \label{lem:H_k-approx}
    Let $\bP$ be a consistent and balanced path profile in $D$. Then $\makespan(\bP,D) \leq H_k \cdot \makespan(\bP^*,D)$.
\end{restatable}

We now introduce the third property, $\phi$-boundedness, which is likewise maintained by greedy composition and implies the second bound on the approximation guarantee when fulfilled by the final solution.

\medskip

\noindent\textbf{$\phi$-boundedness.}
We call a path profile $\bP$ \emph{$\phi$-bounded} in $D' \sqsubseteq D$ if
\begin{equation}
    \label{eq:inv2'}
    \makespan(\bP,D') \leq (\phi(D')+1) \cdot \makespan(\bP^*,D).
\end{equation}
Note that in this definition, the left-hand side refers to the length of the longest path of $\bP$ in the subgraph $D' \sqsubseteq D$, whereas the right-hand side refers to the optimal objective value for \ac{MinMaxDP} on the entire graph $D$.
The following lemma establishes that consistency and $\phi$-boundedness are maintained by the greedy composition when executing all compositions that correspond to one component of the decomposition tree. 
\begin{restatable}{lemma}{GreedyPhi}
    \label{lem:phi-combination}
    Let $D' \sqsubseteq D$ be a subgraph corresponding to a root vertex of an $S$-component or $P$-component, respectively. Let $D^{(1)},\dots, D^{(i)}$ be the graphs corresponding to the children of the component.
    Further, let $\bP^{(j)}$ be a consistent and $\phi$-bounded path profile in $D^{(j)}$ for $j \in [i]$.
    Then, the greedy series, respectively parallel, composition $\bP$ of $\bP^{(1)}, \dots, \bP^{(i)}$ is a consistent and $\phi$-bounded path profile in $D'$.
\end{restatable}
The proof of this lemma involves an induction over the components of the decomposition tree of $D$ and can be found in \Cref{app:lem:phi-combination}.
%the full version of this article\cite{RoutingTrains_ESA-full}.

\subsection{Dynamic programs}
\label{sec:DP}

The results in the preceding section imply that if we start from path profiles for the leaves of the decomposition tree of $D$ that are consistent with an optimal solution and balanced or $\phi$-bounded, respectively, then iteratively applying greedy composition will yield a path profile for $D$ that is consistent and balanced or $\phi$-bounded, respectively.
In particular, such a profile then yields an $H_k$-approximation or $(\phi(D) + 1)$-approximation.
   
However, since we do not know the optimal solution, we do not know which path profiles fulfil these requirements.
A natural approach seems to be to restrict ourselves to a solution with minimum total length (\ie an integral minimum-cost $s$-$t$-flow of value $k$).
However, there is a difficult trade-off between minimising the length of a longest path or minimising the total length of the solution in a subgraph, as the following example shows.

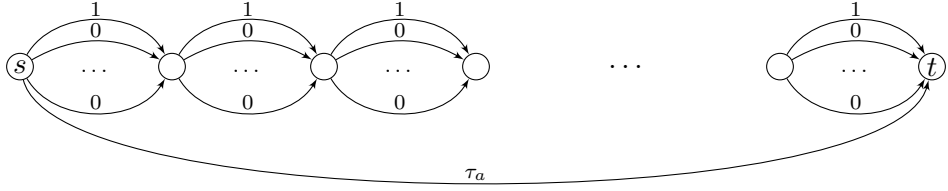
\begin{figure}[tbhp]
    \centering
    \begin{tikzpicture}[scale=1, every node/.style={draw,circle,inner sep=0pt,minimum size=10pt}]
        \node (s) at (0,0) {$s$};
        \node (1) at (2,0) {};
        \node (2) at (4,0) {};
        \node (3) at (6,0) {};
        \node[draw=none] (.) at (8,0) {$\dots$};
        \node (t-) at (10,0) {};
        \node (t) at (12,0) {$t$};

        \path[->, draw] (s) edge[in=120, out=60] node[draw=none, pos=0.5,above=-1pt, font=\scriptsize]{1} (1);
        \path[->, draw] (s) edge[in=150, out=30] node[draw=none, pos=0.5,above=-1pt, font=\scriptsize]{0} (1);
        \path[->] (s) edge[draw=none, in=200, out=-20] node[draw=none, pos=0.5,above=-1pt, font=\scriptsize]{$\dots$} (1);
        \path[->, draw] (s) edge[in=240, out=-60] node[draw=none, pos=0.5,above=-1pt, font=\scriptsize]{0} (1);

        \path[->, draw] (1) edge[in=120, out=60] node[draw=none, pos=0.5,above=-1pt, font=\scriptsize]{1} (2);
        \path[->, draw] (1) edge[in=150, out=30] node[draw=none, pos=0.5,above=-1pt, font=\scriptsize]{0} (2);
        \path[->] (1) edge[draw=none, in=200, out=-20] node[draw=none, pos=0.5,above=-1pt, font=\scriptsize]{$\dots$} (2);
        \path[->, draw] (1) edge[in=240, out=-60] node[draw=none, pos=0.5,above=-1pt, font=\scriptsize]{0} (2);
        
        \path[->, draw] (2) edge[in=120, out=60] node[draw=none, pos=0.5,above=-1pt, font=\scriptsize]{1} (3);
        \path[->, draw] (2) edge[in=150, out=30] node[draw=none, pos=0.5,above=-1pt, font=\scriptsize]{0} (3);
        \path[->] (2) edge[draw=none, in=200, out=-20] node[draw=none, pos=0.5,above=-1pt, font=\scriptsize]{$\dots$} (3);
        \path[->, draw] (2) edge[in=240, out=-60] node[draw=none, pos=0.5,above=-1pt, font=\scriptsize]{0} (3);

        \path[->, draw] (t-) edge[in=120, out=60] node[draw=none, pos=0.5,above=-1pt, font=\scriptsize]{1} (t);
        \path[->, draw] (t-) edge[in=150, out=30] node[draw=none, pos=0.5,above=-1pt, font=\scriptsize]{0} (t);
        \path[->] (t-) edge[draw=none, in=200, out=-20] node[draw=none, pos=0.5,above=-1pt, font=\scriptsize]{$\dots$} (t);
        \path[->, draw] (t-) edge[in=240, out=-60] node[draw=none, pos=0.5,above=-1pt, font=\scriptsize]{0} (t);

        \draw[->] (s) .. controls (0.5,-2) and (11.5,-2) .. (t) node[draw=none, pos=0.5,above=-2pt, font=\scriptsize]{$\tau_a$};
    \end{tikzpicture}
    \caption{The depicted graph $D$ is built as follows.
    Let $D'$ be a graph with $k$ parallel arcs (one arc with length 1 and $k-1$ arcs with length 0). We obtain $D$ by connecting $k$ copies of $D'$ in series plus adding an additional arc from $s$ to $t$ with length $\tau_a$.}
    \label{fig:average_needed}
\end{figure}

\begin{example}
    Consider the graph $D'$ depicted in \Cref{fig:average_needed}.
    If we set $\tau_a \coloneqq k-1$, the path profile with minimal total length consists of $k-1$ paths of length $0$ and one path with length $k-1$, leading to the objective value $k-1$.
    However, the optimal path profile consists of $k$ paths all having length $1$.
    On the other hand, assume that $\tau_a \coloneqq 1+ \varepsilon$ and create a graph $D$ that consists of $k$ copies of $D'$ composed in series.
    In this case, solving the \ac{MinMaxDP} in each subgraph $D'$ results in all paths having length $1$ inside each copy of the example, yielding to the objective value $k$ when combining them in series. 
    However, the optimal path profile uses exactly one arc of length $1+\varepsilon$ in each path and otherwise only length-$0$ arcs, yielding an objective value of $1+\varepsilon$.
    Thus, an algorithm which locally only considers the length of the longest path in the profile, would be  a factor of $k$ apart form the optimal value.
\end{example}

\noindent\textbf{DP framework.}
In the following, we show how to resolve this issue by means of two dynamic programs.
In both cases, our dynamic programming table will have the following attributes per cell:
\begin{itemize}
    \item a subgraph $D' \sqsubseteq D$,
    \item the number of paths $k' \in [k]$ used in $D'$,
    \item the total length $\theta'\in \Theta$ of the path profile in $D'$. 
\end{itemize}
Recall that the total length of a path profile $\bP$ in subgraph $D'$ was defined as the sum of the lengths of the paths in $D'$, \ie\ $\cost(\bP,D') \coloneqq \sum_{P\in \bP} \sum_{a\in A[D']\cap P}\tau_a$.
We let  $\Theta$ denote the set containing all possible values that might appear as total costs.

To fill the cells of the dynamic programming table, we start with the subgraphs consisting of single arcs, \ie\ $D'=(\{v,w\},\{a=(v,w)\})$. Here, the non-empty entries of the table corresponding to $D'$ are $(D',1,\tau_a)=\{(a)\}$ and $(D',0,0)=\{\emptyset\}$, while all remaining entries corresponding to $D'$ are empty (which is different from containing an empty path profile).

Given a cell $(D' \comp D'',\tilde{k},\tilde{\theta})$,  we combine all path profiles
$\bP'$ and $\bP''$ of $D'$ and $D''$ consisting of $k'$ and $k''$ paths of total length $\theta'$ and $\theta''$, respectively, if and only if $\tilde \theta=\theta'+\theta''$ and $\tilde k=k'+k''$ (in case of a parallel composition), or $\tilde k=k'=k''$
(in case of a series composition).
We build path profiles for $D' \comp D''$ according to the greedy composition of the path profiles for $D'$ and $D''$ as described above.

With this procedure, we obtain multiple candidates for a path profile in table entry $(D' \comp D'',\tilde{k},\tilde{\theta})$.
Below, we analyse two different strategies to pick a \emph{good} candidate that lead to different approximation guarantees.
After filling the complete table of the dynamic program, there is a path profile $\bP_{\theta}$ with $k$ paths in $D$ for each possible total-cost value $\theta\in \Theta$.
In the last step, the algorithm will choose the path profile for which the longest path is as short as possible, i.e.\ the algorithm chooses $\arg \min_{\theta \in \Theta} \makespan(\bP_{\theta},D)$.

\medskip

\noindent\textbf{DP strategy for balancedness.}
Our first strategy for selecting the candidates for the intermediate steps aims for balancedness of the solutions and is formalised in the following lemma, that reveals that the strategy yields an $H_k$-approximation.

\begin{restatable}{lemma}{DPHk}
    \label{DPHk}
    The dynamic program which always selects the candidate $\bP$ with $k'$ paths of lengths $p_1 \geq \dots \geq p_{k'}$ that minimises
    \begin{equation}
        \label{eq:DP_choice}
        \max_{i \in \{1, \dots ,k'-1\}} \left\{ \frac{1}{i} \sum_{j=1}^i p_j - p_{i+1} \right\},
    \end{equation}
    computes a solution with $DP \leq H_k \cdot \OPT$, where DP is the objective value of the computed solution and $\OPT$ is the objective value of an optimal solution.
\end{restatable}
The proof of this lemma is by induction on the table entries of the optimal solution and by using \Cref{lem:H_k-combination} to show that we always have a balanced candidate. Together with the choice \eqref{eq:DP_choice}, the lemma follows. A formal proof can be found in \Cref{app:lem:DPHk}.
%the full version of this article~\cite{RoutingTrains_ESA-full}.
\medskip

\noindent\textbf{DP strategy for boundedness.}
Our second candidate selection strategy aims at $\phi$-boundedness of the intermediate solutions. This is formalised in the following lemma, which shows that this strategy yields a $(\phi(D) + 1)$-approximation.

\begin{restatable}{lemma}{DPphi}
    \label{DPphi}
    The dynamic program that always selects the candidate $\bP$ with $k'$ paths of lengths $p_1 \geq \dots \geq p_{k'}$ that minimises
    \begin{alignat}{2}
        &p_1 - p_{k'} \quad &&\text{ in case of a series combination, and } \label{eq:DP_phi_choice_s}\\
        &p_1 &&\text{ in case of a parallel composition} \label{eq:DP_phi_choice_p}
    \end{alignat}
    computes a solution with $DP \leq (\phi(D)+1) \cdot \OPT$, 
    where DP is the objective value of the computed solution and $\OPT$ is the objective value of an optimal solution.
\end{restatable}
Similarly to the proof of \Cref{DPHk} we follow the table entries of an optimal solution. However, here we need to consider a component of the decomposition tree of $D$ at once and show that the entries for the subgraphs corresponding to segment roots are $\phi$-bounded.
The formal proof of this lemma can be found in \Cref{app:lem:DPphi}.
%the full version of this article~\cite{RoutingTrains_ESA-full}.

\medskip

\noindent\textbf{Running time.}
From the size of the dynamic programming table, we can obtain the bound on its running time formalised in the lemma below.
By a standard rounding approach, we can further ensure that the number of possible total-cost values $|\Theta|$ is polynomial in the size of the graph at a loss of a $(1+\varepsilon)$-factor in the approximation ratio, obtaining the desired polynomial running time; see \Cref{app:DP_runtime} for details.
%the full version of this article~\cite{RoutingTrains_ESA-full} for details.
\begin{restatable}{lemma}{DPruntime}
    \label{lem:DPruntime}
    The dynamic program runs in time $\mathcal{O}(m \cdot k^4 \cdot |\Theta|^2)$.
\end{restatable}

\section{Conclusion}

In this paper, we introduced \ac{TMO} as a natural optimisation problem that captures some of the fundamental challenges of routing trains with a headway constraint.
By showing the existence of optimal train routings with a convoy structure, we linked the approximability of the problem to that of \ac{MinMaxDP}, for which we then provided new approximation results in \ac{SePa}-graphs.
Several interesting questions arise from our work:
\begin{itemize}
    \item Is it possible to extend the logarithmic approximation guarantee for \ac{MinMaxDP} in \ac{SePa}-graphs to arbitrary \ac{DAG}s (where it would be best possible due to \cref{thm:hardness})?
    \item Conversely, can the inapproximability result from~\cite{DOKKA23} for \ac{DAG}s be extended to \ac{SePa}-graphs~(again implying that the greedy approach is best possible in those graphs)?
    \item Which of the results presented in this paper can be extended to the setting in which trains might have distinct sources and sinks?
    Note that for general \acp{DAG} even finding disjoint paths for multiple source-sink pairs is \textsf{NP}-hard~\cite{even1976complexity} and that this still holds true in undirected \ac{SePa}-graphs~\cite{nishizeki2001edge}; however, this hardness result does not carry over to directed \ac{SePa}-graphs, so there is hope that our approximation results for those graphs extend to the multi commodity case.
\end{itemize}

%============================================================
\bibliography{routing}
%============================================================

\newpage
\noindent
{\LARGE \textbf{Appendix}}

\appendix

\section{Hardness for MinMaxDP in DAGs -- Proof of \Cref{thm:hardness}}
\label{app:hardness}

In this section, we discuss some implications of a recent inapproximability result by Dokka and Goerigk~\cite{DOKKA23} for the \acf{MLBA} problem, which is defined as follows.

\begin{definition}[\ac{MLBA}]
We are given a directed $b$-partite graph $G=(V_1 \cup \ldots \cup V_b, A)$ where each block $V_i$ has $n$ nodes and each arc is from $V_i$ to $V_{i+1}$, for $i \in [b-1]$. Let $V = \cup_{i \in [b]} V_i$. Each node $v \in V$ has a positive weight $w(v) \in \mathbb{N}$. The problem is to partition the node set $V$ into $n$ node-disjoint paths $P_1, \ldots, P_{ n}$ to minimize the maximum length $\max_{j \in[n]} \sum_{v \in P_j} w(v)$ of any path. 
\end{definition}

It is not hard to see that \ac{MLBA} constitutes a special case of \ac{MinMaxDP}, as described in the proof of the following proposition.

\begin{proposition}\label{prop:reduction-MLBA}
    An $\alpha$-approximation for the \ac{MinMaxDP} problem in DAGs gives an $\alpha$-approximation for \ac{MLBA}.
\end{proposition}

\begin{proof}
Given an instance of \ac{MLBA} $(V_1, \ldots, V_b, E)$, we obtain an instance of \ac{MinMaxDP} as follows. We replace each node $v$ with weight $w(v)$ in the \ac{MLBA} instance by two nodes $v_{in}$, $v_{out}$, and an arc $a_v = (v_{in},v_{out})$ with travel time $\tau_{a_v}=w(v)$. We also replace each arc $a = (u,v)$ in the original instance by an arc $a' = (u_{out},v_{in})$ with travel time $0$. Finally, we add a source $s$ and arcs $(s,v_{in})$ with travel time $0$ for every $v \in V_1$ in the original instance, and a sink $t$ and arcs $(v_{out},t)$ with travel time $0$ for every $v \in V_b$. 
We set $k = n$. 
Clearly, the resulting graph is a \ac{DAG}. Furthermore, any $s$-$t$-path in the \ac{MinMaxDP} instance with length $w$ corresponds to a path of weight $w$ in the \ac{MLBA} instance, and vice versa. 
Moreover, a path is node-disjoint in \ac{MLBA}
if and only if it is arc-disjoint in \ac{MinMaxDP}. Thus, the optimal value remains unchanged in the reduction, and an $\alpha$-approximation for the \ac{MinMaxDP} problem in DAGs gives an $\alpha$-approximation for \ac{MLBA} as well.
\end{proof}

Dokka and Goerigk~\cite{DOKKA23} showed the following inapproximability result for \ac{MLBA}.

\begin{theorem}[Theorem~2.1 in~\cite{DOKKA23}]
\label{thm:dokka-hardness}
    Unless $\textsf{P} = \textsf{NP}$, there is no approximation algorithm for \ac{MLBA} with a ratio strictly better than $b/3+1$, where $b$ is the number of blocks in the \ac{MLBA} instance. 
\end{theorem}

An immediate implication of \cref{thm:dokka-hardness} is that also \ac{MinMaxDP} and \ac{TMO} in \acp{DAG} do not allow for constant-factor approximation algorithms (unless $\textsf{P}=\textsf{NP}$).
We will present an alternative analysis of the construction used in \cite{DOKKA23} to derive the following logarithmic lower bound on the approximation factor in terms of the number of paths $k$.
\hardness*

Dokka and Goerigk~\cite{DOKKA23} establish their hardness-of-approximation result for \ac{MLBA} by reduction from \textsc{3-Dimensional Matching}.
Their construction takes as input an instance of \textsc{3-Dimensional Matching} and a parameter $u$, and constructs an instance of \ac{MLBA} with $b = 3u$ blocks, each containing $n$ nodes where
\begin{equation}
    \label{eq:hardness}
    n \leq u(3q + 6p^{u + 1}),
\end{equation}
where $p$ and $q$, respectively, are the number of hyperedges and elements per dimension (that is, there are $3q$ elements in total) in the \textsc{3-Dimensional Matching} instance~\cite[Lemma~2.1]{DOKKA23}.
The weights in the constructed instance are all in $\{0, 1\}$.
The \ac{MLBA} instance has the property that it has an optimal solution value of $1$ if the \textsc{3-Dimensional Matching} instance admits a feasible solution, and an optimal solution value of $u + 1$ otherwise~\cite[Proposition~2.3]{DOKKA23}.
This immediately implies \cref{thm:dokka-hardness}.

We now observe that the construction described above also implies a strict lower bound of $\log^{1-\varepsilon}(k)$ for \ac{MinMaxDP} in \acp{DAG} for any $\varepsilon > 0$, assuming that \textsf{NP} contains problems that cannot be solved in quasi-polynomial time.

\begin{proof}[Proof of \Cref{thm:hardness}]
Assume there exists a $\log^{1 - \varepsilon}(k)$-approximation for \ac{MinMaxDP} for some $\varepsilon > 0$ and let $c = 1/\varepsilon$.
We show how to use this algorithm to solve \textsc{3-Dimensional Matching} in quasi-polynomial time.
For this, take an instance of \textsc{3-Dimensional Matching} and use the construction from~\cite{DOKKA23} with $u = \lceil (\log p)^c \rceil$ to obtain an \ac{MLBA} instance with $b=3\lceil (\log p)^c \rceil$ blocks and $n \leq u (3q + 6p^{u + 1})$ nodes per block.
Note that the size of the constructed instance is quasi-polynomial in the size $p + q$ of the \textsc{3-Dimensional Matching} instance.
Without loss of generality, we can assume $p \geq q$ (otherwise, the \textsc{3-Dimensional Matching} instance is infeasible).
Note further that for fixed $c = 1/\varepsilon$ there exists a constant $p_c \in \mathbb{N}$ such that both 
\begin{align*}
    ((\log(p))^c +1) \left(3p + 6p^{(\log(p))^c+2} \right) \leq p^{(\log(p))^c +3} \quad \text{ and } \quad (\log p)^{1/c} > 4
\end{align*}
for all $p \geq p_c$.
We can assume that $p \geq p_c$ for the given \textsc{3-Dimensional Matching} instance (otherwise we can solve it by enumeration).
Thus, we obtain
\[
    n \leq u(3q  + 6p^{u+1}) \leq ((\log(p))^c +1) \left(3p + 6p^{(\log(p))^c+2} \right) \leq p^{(\log(p))^c +3},
\]
where the first inequality follows from \eqref{eq:hardness}, the second follows from definition of $u$ and $p \geq q$, and the final inequality follows from $p \geq p_c$.

We apply the $\log^{1 - \varepsilon}(k)$-approximation for \ac{MinMaxDP} (using the reduction sketched in \cref{prop:reduction-MLBA}) to get a solution to the \ac{MLBA} instance whose value is at most $\log^{1 - \varepsilon}(n)$ times the optimal solution value (note that $n$ takes the role of $k$ when interpreting \ac{MLBA} as a special case of \ac{MinMaxDP}).
If the \textsc{3-Dimensional Matching} instance is feasible, the optimal solution value of \ac{MLBA} is $1$ and we obtain a solution of cost at most $\log^{1 - 1/c}(p^{(\log p)^c + 3})$.
Observe that
\begin{align*}
    \left( \log\left(p^{(\log p)^c + 3}\right) \right)^{1 - 1/c} & = \left( (\log p)^{c + 1} + 3 \log p \right)^{1 - 1/c} \\
    &\leq \left( (\log p)^{c + 1 + 1/c} \right)^{1 - 1/c} 
    = \left( \log p \right)^{c - (\frac{1}{c})^2}
    \leq (\log p)^c < u + 1,
\end{align*}
where the first inequality follows from our assumption that $(\log p)^{1/c} > 4$.
Thus, we can distinguish between instances with optimal value $1$ or $u + 1$ arising from the construction and solve \textsc{3-Dimensional Matching} in quasi-polynomial time.
\end{proof}

\section{Existence of optimal convoy routings -- Proof of \Cref{thm:ConvoyConjecture}}
\label{app:ConvoySolution}

In this section, we formally prove the convoy theorem.
\ConvoySolution*
As mentioned in~\Cref{sec:ConvoySolutions}, the proof of this theorem is algorithmic, via the two-phase \textit{uncrossing algorithm}.
In the following, we formally describe the uncrossing algorithm and give a complete proof of the theorem.

We assume we are given an optimal routing $(\bP,\blam)$ where each path $P \in \bP$ is acyclic. Recall the definition of leaders $L_a$ for $a \in \delta^-(t)$ which is the first train using arc $a$, as well as the sets $L= \bigcup_{a \in \delta^-(t)}L_a$, and $S_a$ for each $a \in A$ which is the set of trains using arc $a$ (see \Cref{sec:ConvoySolutions}).
For trains $x,y \in S_a$, we write $x \prec_a y$ ($x$ precedes $y$ on arc $a$) to denote that $x$ uses arc $a$ before $y$ (thus, $\lambda_x(a) < \lambda_y(a)$).
Similarly, we write $x \succ_a y$ ($x$ succeeds or follows $y$ on arc $a$) if $x$ uses $a$ after train $y$. Let $\bar{S}_a$ be the \emph{ordered tuple} of trains using arc $a$. Thus if $x \prec_a y$, then $x$ appears before $y$ in $\bar{S}_a$.
For a train $x \in S_a$, define $\bar{S}^x_a$ be the ordered tuple of trains that use arc $a$ \emph{before} train $x$ (and including train $x$).
For train $x$, we say arcs $a$, $a' \in \Path{x}$ are \emph{congruent} if $\bar{S}_a^x = \bar{S}_{a'}^x$, \ie\ the trains preceding $x$ on $a$ also precede $x$ on $a'$, in the same order.
We overload notation and use $a \prec_x a'$ to denote that $a$ precedes $a'$ in path $\Path{x}$. The operators $\Prec_x(a)$ and $\Succ_x(a)$ denote the arcs just before and just after arc $a$ on the path $\Path{x}$.

For a train $x$, let $\ell_x \in \delta^-(t)$ be the last arc in $\Path{x}$. An important definition for each train in our algorithm is the \emph{transition arc} $a^{tr}_x$: for train $x$, the transition arc is the last arc $a \in \Path{x}$ that is not congruent with $\ell_x$, \ie\ for which $\bar{S}^x_a \neq \bar{S}^x_{\ell_x}$. The \emph{transition node} $v^{tr}_x$ is the node at the head of the arc $a^{tr}_x$. \Cref{fig:conveyexample} depicts an example, showing a transition arc and transition node. Thus $a^{tr}_x = \Path{x} \cap \delta^-(v^{tr}_x)$. By definition, $a^{tr}_x \neq \ell_x$. Further, all arcs following $a^{tr}_x$ on $\Path{x}$ are congruent for $x$ -- the same trains precede $x$ in the same order. If $\bar{S}^x_a = \bar{S}^x_{\ell_x}$ on all arcs in $\Path{x}$, \ie\ all arcs in $\Path{x}$ are congruent for train $x$, we define $a^{tr}_x = v^{tr}_x = \bot$. In this case, we say train $x$ \emph{moves in a convoy}. Define $Q_x$ as the sub-path of $\Path{x}$ from $v^{tr}_x$ to $t$. In particular, $a^{tr}_x \not \in Q_x$, and all arcs in $Q_x$ are congruent for $x$. Note that if train $x$ is a leader, then on all arcs in $Q_x$, train $x$ is the first train.

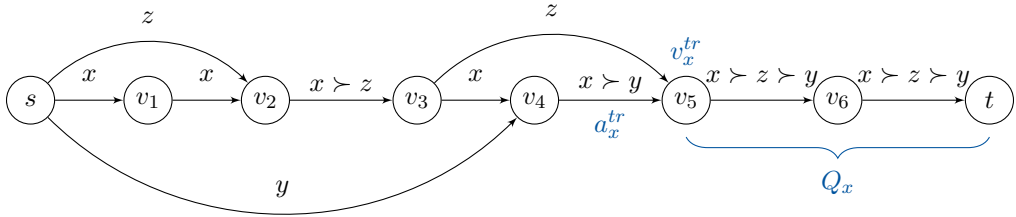
\begin{figure}[htbp]
\centering
\begin{tikzpicture}[scale=0.9, transform shape, every node/.style={inner sep=0pt,minimum size=20pt, draw, circle, align=center}, 
    edge from parent/.style={draw, thick, -{Circle[open]}, font=\small}]
    % Vertices
    \node (s) {$s$};
    \node (v1) [right=of s] {$v_1$};
    \node (v2) [right=of v1] {$v_2$};
    \node (v3) [right=of v2, xshift=0.5cm] {$v_3$};
    \node (v4) [right=of v3] {$v_4$};
    \node (v5) [right=of v4, xshift=0.5cm] {$v_5$};
    \node (v6) [right=of v5, xshift=0.5cm] {$v_6$};
    \node (trv) [above=of v5, yshift=-1cm, draw=none, fill=none, color=rwthblue] {$v^{tr}_x$};
    \node (t) [right=of v6, xshift=0.5cm] {$t$};
    
    % Path edges (labeled "x")
    \draw[->] (s) -- (v1) node[draw=none,fill=none,midway,above] {$x$};
    \draw[->] (v1) -- (v2) node[draw=none,fill=none,midway,above] {$x$};
    \draw[->] (v2) -- (v3) node[draw=none,fill=none,midway,above,yshift=-0.25cm] {$x \succ z$};
    \draw[->] (v3) -- (v4) node[draw=none,fill=none,midway,above] {$x$};
    \draw[->] (v4) -- (v5) node[draw=none,fill=none,midway,above,yshift=-0.25cm] {$x \succ y$} node[fill=none,draw=none,midway,below,color=rwthblue] {$a^{tr}_x$};
    \draw[->] (v5) -- (v6) node[draw=none,fill=none,midway,above,yshift=-0.5cm] {$x \succ z \succ y$};
    \draw[->] (v6) -- (t) node[draw=none,fill=none,midway,above,yshift=-0.5cm] {$x \succ z \succ y$};

    % Custom labeled curved edges (no box around labels)
    \draw[->, bend left=45] (s) to node[draw=none, fill=none, above, inner sep=1pt] {$z$} (v2);
    \draw[->, bend right=45] (s) to node[draw=none, fill=none, above, inner sep=1pt] {$y$} (v4);
    \draw[->, bend left=45] (v3) to node[draw=none, fill=none, above, inner sep=1pt] {$z$} (v5);
    
    % Curly brace
    \draw [rwthblue,decorate,decoration={brace,amplitude=8pt,mirror,raise=1ex}] (v5.south) -- (t.south) node[draw=none, fill=none, midway,below,yshift=-0.5cm,color=rwthblue] {$Q_x$};
\end{tikzpicture}
\caption{An example of a train routing, showing the transition arc $a^{tr}_x$, the transition node $v^{tr}_x$, and the segment $Q_x$ for train $x$. In this routing, $x$ is not a leader, since $y$ and $z$ precede $i$ on the last arc. In fact, train $y$ is a leader, and $a^{tr}_y = \bot$, since it is the first train on every arc in $\Path{y}$.}
\label{fig:conveyexample}
\end{figure}

The algorithm proceeds as follows.
Given an optimal routing $(\bP,\blam)$, the objective in the first phase is to obtain another optimal routing where the set of leaders $L$ all follow arc-disjoint paths.
If $a^{tr}_x = \bot$ for all leaders $x \in L$, then each leader $x$ is the first train on every arc in its path, and thus all leaders must already be using arc-disjoint paths.
If not, let $x \in L$ be such that $a^{tr}_x \neq \bot$. The algorithm then reroutes all the trains on $a^{tr}_x$ that precede $x$, in this order, to use the path $Q_x$. Prior to the rerouting, $x$ was the first train on each arc in $Q_x$. After the rerouting, $x$ is preceded by all the trains in $\bar{S}^x_{a^{tr}_x}$. Phase 1 repeats this process --- picking a leader $x$ with $a^{tr}_x \neq \bot$, and rerouting all trains in $\bar{S}^x_{a^{tr}_x}$ onto the path $Q_x$ --- until $a^{tr}_x = \bot$ for each leader $x \in L$. Note that in this process, the new transition arc $a^{tr}_x$ for train $x$ moves closer to $s$. We use this in~\Cref{lem:uncrossingruntime} to prove that the algorithm for this phase terminates.

In Phase 2, we simply take each arc $a \in \delta^-(t)$, and re-route all trains using this arc to follow the path of the leader on this arc. We show that this gives us the necessary routing.

The proof of~\Cref{thm:ConvoyConjecture} then follows from~\Cref{lem:convoyyes} (which shows that if~\Cref{alg:uncrossing} terminates, it returns an optimal convoy routing) and~\Cref{lem:uncrossingruntime} (which shows that~\Cref{alg:uncrossing} indeed terminates in finite time).
\begin{algorithm}[h]
    \algorithmInit
    \KwIn{An optimal routing $(\bP,\blam)$}
    \KwOut{An optimal convoy routing $(\bR, \bsigma)$}
    \tcp{Phase 1: Re-route leaders to use disjoint paths}
    $L \coloneqq \{x \in [d] \mid \exists\ a \in \delta^-(t) \text{ with } x \in S_a \text{ and } \forall y \in S_a, \, x \preceq_a y\}$ \tcp*{$L$ is the set of leaders.}
    \While{$\exists\ x \in L \text{ for which } a^{tr}_x \neq \bot$} {
        \tcp{Re-route all trains preceding $x$ on $a^{tr}_x$ to use the path $Q_x$, in the same order}
        \For{$y \in \bar{S}^x_{a^{tr}_x}$}{
            %Let $R_1$ be the $s$-$v^{tr}_i$ segment of $\Path{j}$, and $R_2$ be the $v^{tr}_i$-$t$ segment. \\
            $\Path{y} \coloneqq (\Path{y}[s,v_x^{tr}],\ Q_x)$ \\
            \tcp{Adapt the entry functions $\lambda_y$}
            $a \coloneqq a^{tr}_x$ \tcp*{$\lambda_y(a)$ is unchanged for arcs $a \in \Path{y}[s,v^{tr}_x]$} 
            \While{$\exists \, \Succ_x(a)$} {
                $a' \coloneqq \Succ_x(a)$,\\
                $\lambda_y(a') \coloneqq \lambda_y(a) + \tau_a$,\\
                $a \coloneqq a'$
            }
        }
        \tcp{Each train $y$ is rerouted onto $Q_x$, and follows the same order as on $a^{tr}_x$}
		Update the set of leaders $L$ \\
    }
    \tcp{Phase 2: Re-route all trains to follow their leaders}
    $\textrm{count} = 1$\\
    \For{\text{each leader} $x \in L$}{
        $R_{\textrm{count}} \coloneqq \Path{x}$, $\sigma_{\textrm{count}} \coloneqq |S_{\ell_x}|$ \tcp*{All trains using arc $\ell_x$ will use $\Path{x}$ in $(\bR,\bsigma)$} 
        $\textrm{count} \coloneqq \textrm{count}+1$\\
    }
    \Return convoy routing $(\bR,\bsigma)$
	\caption{Uncrossing algorithm to transform an optimal solution to a convoy routing}
    \label{alg:uncrossing} 
\end{algorithm}

\begin{lemma}
    If~\Cref{alg:uncrossing} terminates, it returns an optimal routing for the given instance where all trains travel in convoys.
    \label{lem:convoyyes}
\end{lemma}
\begin{proof}
    We will show that the routing obtained after each phase remains feasible, and that the makespan does not increase. It follows that if the algorithm terminates, it is an optimal routing where the trains travel in convoys.

    Consider first a single iteration of the while loop in Phase 1. Let $x \in L$ be the train picked by the algorithm.
    Then prior to re-routing the trains preceding $x$ in the for loop, train $x$ was the first train on each arc in $Q_x$. In the re-routing step, each train $y$ that preceded $x$ on $a^{tr}_x$ is rerouted onto the path $Q_x$, in the same order. Further for any trains $y$, $y'$ that are rerouted and any arc $a \in Q_x$, $\lambda_y(a) - \lambda_{y'}(a) = \lambda_y(a^{tr}_x) - \lambda_{y'}(a^{tr}_x)$. The algorithm thus maintains the gap between rerouted trains. Clearly, this is a feasible routing. Further, all the rerouted trains arrive at $t$ before train $x$, and hence the makespan does not increase in Phase 1.

    In Phase 2, consider the makespan prior to the phase. On an arc $a \in \delta^-(t)$, if $x$ is the leader, then $x$ arrives at time at least $\sum_{a \in \Path{x}} \tau_a$. Further since $|S_a|$ trains use this arc to enter $t$, the last train enters $t$ at time at least $\Delta \cdot (|S_a|-1)$ after train $x$. Thus the last train enters $t$ at time at least $\sum_{a \in \Path{x}} \tau_a + \Delta \cdot (|S_a|-1)$.
    
    All trains in $S_a$ are rerouted to use the path $\Path{x}$, and leave $s$ at times $0, \Delta, \dots, \Delta \cdot (|S_a|-1)$. Clearly, all of these trains arrive at $t$ by time $\sum_{a \in \Path{x}} \tau_a + \Delta \cdot (|S_a|-1)$. Hence, the makespan does not increase in Phase 2. Since prior to Phase 2 the leaders use disjoint paths, after rerouting the trains travel in convoys, completing the proof.
\end{proof}

It remains to show that Phase 1 in the algorithm indeed terminates. To show this, we first need two propositions on the paths used by trains. Throughout the remainder of this section, we fix a particular iteration of the while loop in which train $x$ is picked. The routing $(\bP,\blam)$ and associated variables $Q_x$, $e^{tr}_x$, $v^{tr}_x$ etc. refer to the routing before the trains preceding $x$ on $a^{tr}_x$ are rerouted, \ie\ before the execution of the for loop. We will use the superscript ``$\new$'' to refer to quantities following the termination of the for loop, after all the trains preceding $x$ on $a^{tr}_x$ are rerouted onto $Q_x$.

\begin{proposition}
    \label{prop:uncrossing_circle}
    Let $y$ be the train picked for rerouting in the for loop in Phase 1. Then $\Path{y}[s,v^{tr}_x]$ and $Q_x$ are arc-disjoint.    
\end{proposition}

\begin{proof}
    Suppose for a contradiction that some arc $a \in \Path{y}[s,v^{tr}_x] \cap Q_x$.
    Then since $a \in \Path{y}[s,v^{tr}_x]$, we know that $\lambda_y(a) \leq \lambda_y(a^{tr}_x)$, \ie\ train $y$ uses arc $a$ before it uses arc $a^{tr}_x$.
    Since $y$ precedes $x$ on arc $a^{tr}_x$, $\lambda_y(a^{tr}_x) < \lambda_x(a^{tr}_x)$.
    Further, since $a \in Q_x$, $\lambda_x(a^{tr}_x) \leq \lambda_x(a)$. Putting these together, it must be true that $\lambda_y(a) < \lambda_x(a)$. On every arc in $Q_x$, however, train $x$ is the first train, and hence $\lambda_x(a) < \lambda_y(a)$, giving a contradiction.
\end{proof}

\begin{proposition}
    \label{prop:conflict_free}
    For trains $y$ and $z$, if $Q_y \cap Q_z \neq \emptyset$, then either $Q_y \subseteq Q_z$ or $Q_y \supseteq Q_z$. Further, if $y$ is a leader and $Q_y \cap Q_z \neq \emptyset$, then $Q_y \supseteq Q_z$. 
\end{proposition}

\begin{proof}
    Suppose arc $a \in Q_y \cap Q_z$, and assume \WLOG that train $y$ precedes train $z$ on arc $a$, \ie\ $\lambda_y(a) < \lambda_z(a)$.
    We will then show that $Q_z \subseteq Q_y$.
    By definition of $Q_z$, and since $y$ precedes $z$ on $a$, train $y$ must precede $z$ on every arc in $Q_z$, and hence $Q_z \subseteq \Path{y}$. In fact, since $t$ is the last node in $Q_z$, $Q_z$ is a suffix of $\Path{y}$. Since $Q_y$, $Q_z$ are both suffixes of $\Path{y}$, either $Q_y \subset Q_z$ if $y$'s transition arc $a^{tr}_y \in Q_z$, or $Q_z \subseteq Q_y$ otherwise. We will show that $a^{tr}_y \not \in Q_z$, and hence $Q_z \subseteq Q_y$.

    Let $a' = \Succ_y(a^{tr}_y) $.
    Then $a^{tr}_y$ and $a'$ are not congruent for $y$: either the order or the set of trains preceding $y$ differs on $a^{tr}_y$ and $a'$.
    Suppose for a contradiction that $a^{tr}_y \in Q_z$, then $a' \in Q_z$ as well.
    But then since $y$ precedes $z$ on $Q_z$, that implies that also the order or the set of trains before $z$ changes, \ie $a_y^{tr}$ and $a'$ are not congruent for $z$. However, $a_y^{tr}, a' \in Q_z$, which contradicts the definition of $Q_z$.
    Hence, $y$'s transition arc $a^{tr}_y \not \in  Q_z$, and hence $Q_z \subseteq Q_y$.

    Note that if $y$ is a leader, then $y$ precedes $z$ on $a \in Q_y \cap Q_z$. This means the assumption just taken must in fact be true in this case, and hence $Q_z \subseteq Q_y$.
\end{proof}

We now show the following theorem mainly to prove that  Algorithm~\ref{alg:uncrossing} terminates. The given runtime is not polynomial. However, we just use the existence result not the actual algorithm in the following. 

\begin{lemma}
    \label{lem:uncrossingruntime}
    Algorithm~\ref{alg:uncrossing} terminates in $O(mnd^2)$ time, where $d$ is the number of trains.
\end{lemma}

\begin{proof}
    Clearly, Phase 2 terminates in $|L|$ iterations. Each leader is the first train on some arc $a \in \delta^-(t)$, and hence $|L| \le m$. 

    To show the bound for the while loop in Phase 1, we claim that in each iteration of the while loop, if $x$ is the train picked, the number of arcs in the path $\Path{x}[s,v^{tr}_x]$ strictly decreases.
    Further, for any other train $y$, the number of arcs in $\Path{y}[s,v^{tr}_y]$ does not increase.
    Since there are $d$ trains and each path has length at most $n-1$ (since we start with an acylic optimal routing), the bound follows.

    Consider first the leader $x$, chosen in an iteration of the while loop. Then in the while loop, all trains preceding $x$ on $a^{tr}_x$ are re-routed to use $Q_x$, in the same order. Thus $a^{tr}_x$ is now congruent for $x$ with the edges in $Q_x$, and hence the new transition arc $a^{tr,\new}_x \prec_x a^{tr}_x$. This shows the claim for train $x$.

    Next, consider any train $y$ that is rerouted in the for loop. By the algorithm, prior to the for loop, all trains that precede $y$ on $a^{tr}_x$, will also precede $y$ in the same order on all arcs in $Q_x$ when the for loop terminates, \ie\ arcs $a^{tr}_x$ and $Q_x$ are congruent for $y$ when the for loop terminates.
    Hence, $a^{tr,\new}_y \prec_y a^{tr}_x$.
    If $a^{tr}_x \preceq_y a^{tr}_y$, then clearly the claim holds for $y$ as well.
    So assume $a^{tr}_y \prec_y a^{tr}_x$.
    Thus, the arc $a^{tr}_x \in Q_y$.
    \Cref{fig:reroutedtrain} depicts an example of this case. 

    Thus after the for loop, $a^{tr}_x$ and all arcs on $\Path{y,\new}$ after $a^{tr}_x$, are congruent for train $y$. We need to show this is true for the segment between $a^{tr}_y$ and $a^{tr}_x$ as well, \ie\ all arcs in this segment are congruent with $a^{tr}_x$ for $y$. For this, let $a$ be any arc such that $a^{tr}_y \prec_y a \preceq_y a^{tr}_x$. We claim that after the for loop,  arcs $a$ and $a^{tr}_x$ are congruent for train $y$, \ie
    \begin{align} \label{eqn:rerouted}
        \bar{S}^{y,\new}_a = \bar{S}^{y,\new}_{a^{tr}_x} \, .
    \end{align}

    If \eqref{eqn:rerouted} holds, then after the for loop, all arcs following $a^{tr}_y$ are congruent for $y$, and hence $a^{tr,\new}_y \preceq_y a^{tr}_y$.

    To see that \eqref{eqn:rerouted} holds, note that the sequence of trains on $a^{tr}_x$ is unchanged, and in particular, $\bar{S}^{y,\new}_{a^{tr}_x} =  \bar{S}^{y}_{a^{tr}_x}$.
    In turn, since both $a$ and $a^{tr}_x$ are arcs in $Q_y$, these arcs are congruent for $y$.
    Thus $\bar{S}^{y,\new}_{a^{tr}_x} = \bar{S}^y_a$.
    To prove~\Cref{eqn:rerouted}, we thus need to show that $\bar{S}^{y,\new}_a = \bar{S}^y_a$.

    Any train $z$ that precedes $y$ on $a$, precedes $y$ on $a^{tr}_x$ as well.
    Hence, it gets rerouted.
    By \Cref{prop:uncrossing_circle}, $\Path{z}[s,v^{tr}_x] \cap Q_x = \emptyset$, so $\Path{z}[s, v^{tr}_x]$ is not changed.
    Thus, these trains continue to precede $y$ on $a$ in the same order.

    Since $a \not \in Q_x$, and trains only get re-routed onto $Q_x$, no train that previously did not use $a$, uses arc $a$ after the rerouting. Thus, on arc $a$, $\bar{S}^{y,\new}_a = \bar{S}^y_a$, completing the proof of \eqref{eqn:rerouted} and thus $a^{tr,\new}_y \preceq_y a^{tr}_y$.

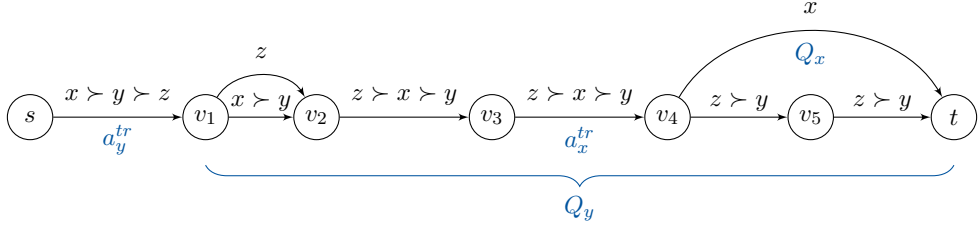
\begin{figure}[t]
\centering
\begin{tikzpicture}[scale=0.85, transform shape,
    every node/.style={draw, circle, align=center,inner sep=0pt,minimum size=20pt}, 
    edge from parent/.style={draw, thick, -{Circle[open]}, font=\small}]
    % Vertices
    \node (s) {$s$};
    \node (v1) [right=of s, xshift=1cm] {$v_1$};
    \node (v2) [right=of v1] {$v_2$};
    \node (v3) [right=of v2, xshift=1cm] {$v_3$};
    \node (v4) [right=of v3, xshift=1cm] {$v_4$};
    \node (v5) [right=of v4, xshift=0.5cm] {$v_5$};
    \node (t) [right=of v5, xshift=0.5cm] {$t$};

    % Path edges (labeled "i")
    \draw[->] (s) -- (v1) node[draw=none,fill=none,midway,above,yshift=-0.5cm] {$x \succ y \succ z$}  node[fill=none,draw=none,midway,below,color=rwthblue] {$a^{tr}_y$};
    \draw[->] (v1) -- (v2) node[draw=none,fill=none,midway,above,yshift=-0.25cm] {$x \succ y$};
    \draw[->] (v2) -- (v3) node[draw=none,fill=none,midway,above,yshift=-0.5cm] {$z \succ x \succ y$};
    \draw[->] (v3) -- (v4) node[draw=none,fill=none,midway,above,yshift=-0.5cm] {$z \succ x \succ y$} node[fill=none,draw=none,midway,below,color=rwthblue] {$a^{tr}_x$};
    \draw[->] (v4) -- (v5) node[draw=none,fill=none,midway,above,yshift=-0.25cm] {$z \succ y$} ;
    \draw[->] (v5) -- (t) node[draw=none,fill=none,midway,above,yshift=-0.25cm] {$z \succ y$};

    % Custom labeled curved edges (no box around labels)
    \draw[->, bend left=60] (v1) to node[draw=none, fill=none, above, inner sep=1pt] {$z$} (v2);
    \draw[->, bend left=60] (v4) to node[draw=none, fill=none, above, inner sep=1pt] {$x$} node[draw=none, fill=none, below, inner sep=1pt,color=rwthblue] {$Q_x$} (t);
    
    % Curly brace
    \draw [rwthblue,decorate,decoration={brace,amplitude=8pt,mirror,raise=2ex}] (v1.south) -- (t.south) node[draw=none, fill=none, midway,below,yshift=-0.7cm,color=rwthblue] {$Q_y$};
\end{tikzpicture}
\caption{An example where $a^{tr}_y \prec_y a^{tr}_x$ for a rerouted train $y$.}
\label{fig:reroutedtrain}
\end{figure}

Finally, consider a train $z \neq x$ that is not rerouted in the for loop. Note that $x$ is a leader in the routing before the for loop, and hence by~\Cref{prop:conflict_free}, $Q_z$ and $Q_x$ are either arc-disjoint, or $Q_z \subseteq Q_x$. We consider these two cases separately. \Cref{fig:reroutedtraink1} depicts the first case, where $Q_z \subseteq Q_x$.

\begin{figure}[t]
\centering
\begin{tikzpicture}[scale=0.85, transform shape, every node/.style={draw, circle, align=center,inner sep=0pt,minimum size=20pt}, 
    edge from parent/.style={draw, thick, -{Circle[open]}, font=\small}]

    % Vertices
    \node (s) {$s$};
    \node (v1) [right=of s] {$v_1$};
    \node (v2) [right=of v1] {$v_2$};
    \node (v3) [right=of v2, xshift=1cm] {$v_3$};
    \node (v4) [right=of v3, xshift=1cm] {$v_4$};
    \node (t) [right=of v4, xshift=1cm] {$t$};

    % Path edges (labeled "i")
    \draw[->] (s) -- (v1) node[draw=none,fill=none,midway,below,yshift=0.25cm] {$x \succ y$} node[fill=none,draw=none,midway,above,color=rwthblue] {$a^{tr}_x$};
    \draw[->] (v1) -- (v2) node[draw=none,fill=none,midway,below] {$x$} ;
    \draw[->] (v2) -- (v3) node[draw=none,fill=none,midway,below,yshift=0.25cm] {$y \succ x$};
    \draw[->] (v3) -- (v4) node[draw=none,fill=none,midway,below,yshift=0.5cm] {$z \succ y \succ x$};
    \draw[->] (v4) -- (t) node[draw=none,fill=none,midway,below,yshift=0.5cm] {$z \succ y \succ x$} ;

    % Custom labeled curved edges (no box around labels)
    \draw[->, bend left=60] (s) to node[draw=none, fill=none, above, inner sep=1pt] {$z$} node[draw=none, fill=none, below, inner sep=1pt,color=rwthblue] {$a^{tr}_z$} (v3);
    \draw[->, bend left=60] (v1) to node[draw=none, fill=none, below, inner sep=1pt] {$y$} (v2);
    
    % Curly brace
    \draw [rwthblue,decorate,decoration={brace,amplitude=8pt,mirror,raise=2ex}] (v3.south) -- (t.south) node[draw=none, fill=none, midway,below,yshift=-0.7cm,color=rwthblue] {$Q_z$};
    \draw [rwthblue,decorate,decoration={brace,amplitude=8pt,raise=3ex}] (v1.north) -- (t.north) node[draw=none, fill=none, midway,above,color=rwthblue,yshift=0.8cm] {$Q_x$};
\end{tikzpicture}
\caption{An example where $Q_z \subseteq Q_x$. Since $y$ precedes $x$ on $a^{tr}_x$, after rerouting, the order will be $z \succ x \succ y$ on both arcs in $Q_z$.}
\label{fig:reroutedtraink1}
\end{figure}
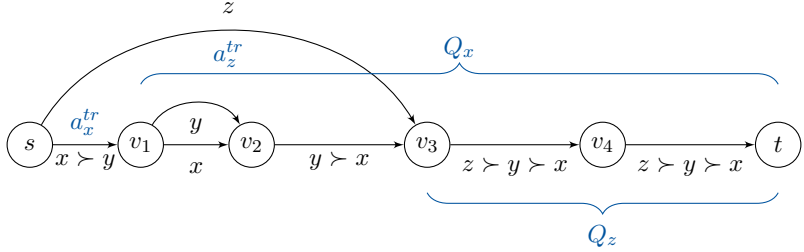

\begin{itemize}
    \item Case 1: $Q_z \subseteq Q_x$. Then prior to the rerouting, $x$ was the first train on all arcs in $Q_z$.
    Let $a \in Q_z$.
    Then $\bar{S}^z_a$ is the ordered tuple of trains preceding $z$ on $a$ (including $z$).
    Note that $x \in \bar{S}_a^z$. Some of these trains may precede train $x$ on arc $a^{tr}_x$, and hence may be rerouted. Let $T$ be the set of trains that are rerouted in the for loop. Following the rerouting, the rerouted trains proceed in the same order on all arcs of $Q_z$, and precede $x$. The trains in $\bar{S}_a^z \setminus T$ are not rerouted, and hence maintain the same order as before the rerouting. Train $x$ precedes the other trains in $\bar{S}_a^z \setminus T$, and is preceded by the rerouted trains. Overall, the arcs of $Q_z$ remain congruent for train $z$ after rerouting, and hence $a^{tr,\new}_z \prec_z a^{tr}_z$.
    
    \item Case 2: $Q_x$ and $Q_z$ are arc-disjoint. Again, let $a \in Q_z$, and $\bar{S}^z_a$ is the ordered tuple of trains preceding $z$ on $a$ (including train $z$).
    Some of these trains may precede train $x$ on arc $a^{tr}_x$, and hence may be rerouted. Let $T$ be the set of trains that are rerouted in the for loop. We will show that on $Q_z$, following the rerouting in the for loop, exactly the remaining trains $\bar{S}^z_a \setminus T$ precede $z$, and in the same order as prior to the rerouting. Thus, again, $a^{tr,\new}_z \prec_z a^{tr}_z$.
    
    Firstly, since all re-routed trains use arcs in $Q_x$ after the rerouting, there are no new trains on the arcs of $Q_z$ after the rerouting. Thus we just need to show that the trains in $T$ that preceded $z$ on $Q_z$, do not use any arc in $Q_z$ after rerouting. 

    Note first that $x$'s transition arc $a^{tr}_x \not \in Q_z$. To see this, assume for a contradiction that $a^{tr}_x \in Q_z$.
    If $z$ precedes $x$ on $a^{tr}_x$, then train $z$ must be rerouted in the while loop. If $x$ precedes $z$ on $a^{tr}_x$, then $Q_z$ has an arc on which $x$ precedes $z$, and an arc which $x$ does not use (the last arc of $Q_z$, since $Q_z$ and $Q_x$ are arc-disjoint). This clearly violates the definition of $Q_z$. Thus, $a^{tr}_x \not \in Q_z$.

    Now consider a train $y \in T$ that preceded $z$ on $Q_z$ before being rerouted. Suppose for a contradiction that $y$ uses an arc in $Q_z$ after being rerouted. Let $a$ be the last arc in $Q_z$ used by $y$ after rerouting. Since $Q_z$ and $Q_x$ are arc-disjoint, $a$ cannot be the last arc in $Q_z$. Then let $a' = \Succ_z(a)$ be the arc immediately following $a$ in $Q_z$. But then since $y$ uses arc $a'$ prior to the rerouting, but not after it, and $a,a' \not \in Q_x$, the only possibility is that arc $a = a^{tr}_x$. But then $a^{tr}_x \in Q_z$, giving a contradiction.
    \qedhere
\end{itemize}
\end{proof}

Together,~\Cref{lem:convoyyes} and~\Cref{lem:uncrossingruntime} complete the proof of~\Cref{thm:ConvoyConjecture}.

%------------------------------------------------------------------------------------
\section{A flow-based additive approximation for TMO – Proof of \Cref{thm:FoT-APX}}
\label{app:FlowAlgo}

In this section, we formally prove \Cref{thm:FoT-APX} by describing how to compute a convoy routing for a given \ac{TMO} instance with makespan at most $\OPT + \Delta$ in polynomial time.
\FoTapprox*

To show this theorem, we make use of the connection between convoy routings and temporally repeated flows in \Cref{lem:routing-flow} and \Cref{lem:flow-routing} below.

\begin{lemma}
	\label{lem:routing-flow}
	Let $\OPT$ be the optimal makespan for \ac{TMO} on $(D,\tau,d,\Delta)$. Then, there exists a flow over time $f$ with time horizon $\OPT+\Delta$ and value $d \cdot \Delta$ in graph $D$ with unit capacities.
\end{lemma}

\begin{proof}
	Let $(\bP^*, \bsigma^*)$ be an optimal convoy routing and let $f$ be the temporally repeated flow obtained by sending flow at a rate of one unit per time unit along $P_i^* \in \bP^*$ for $(\OPT+\Delta)- \tau(P_i^*)$ time units.
    Note that $f$ obeys unit capacities as $\bP^*$ is a convoy routing, \ie the paths are arc-disjoint.
	As $\sigma_i^*$ trains traverse $P_i^*$ until time $\OPT$, the length of $P_i^*$ is bounded by $\OPT - (\sigma_i^* -1) \cdot \Delta$, since the trains need to respect the headway.
    Hence, $f$ sends at least $\sigma_i \cdot \Delta$ flow units  along $P_i^*$ as the following estimation shows
	\[
		(\OPT+\Delta)- \tau(P_i^*) \geq (\OPT+\Delta) - (\OPT - (\sigma_i^* -1) \cdot \Delta) = \sigma_i^* \cdot \Delta.
	\]
	Summing this up for all $P_i^* \in \bP^*$, we get
	\[
		\val(f) \geq \sum_{i} \sigma_i^* \cdot \Delta = d \cdot \Delta. \qedhere
	\]  
\end{proof}

\begin{lemma}
	\label{lem:flow-routing}
	Let $f$ be a flow over time in $D$ with time horizon $T$ and unit capacities, which sends $\val(f)$ flow. There exists a strongly polynomial time algorithm which computes a convoy routing with makespan $T$ and value $\frac{1}{\Delta} \cdot \val(f)$.
\end{lemma}

\begin{proof}
	Let $f^*$ be an optimal, integral temporally repeated flow for time horizon $T+1$. Such a flow can be computed in strongly polynomial time (see \cite{ahuja1988network, Skutella2009}).
	Then, $\val(f^*)\geq \val(f)$ as there exists a temporally repeated flow that has maximum value \cite{FordFulkerson1958,ford2015flows}.
	
	Let $P_1,\dots,P_k$ be the (simple) paths along which flow is sent in $f^*$.
	Note that these paths are pairwise arc disjoint since we have unit capacities. Moreover, the aforementioned algorithm for computing $f^*$ also computes such a path decomposition of $f^*$ in strongly polynomial time.
	We construct a convoy routing $(\bP, \bsigma)$ by sending
	\[
		\sigma_i \coloneqq 1 + \left\lfloor \frac{T - \tau(P_i)}{\Delta} \right\rfloor
	\]
	trains through $P_i$, \ie $\bP=(P_1,\dots,P_k)$ and $\bsigma=(\sigma_1,\dots,\sigma_k)$.
	Note that all trains arrive until time horizon $T$. This is even the case if $\tau(P_i)=T$, as we say a train arrives as soon as its nose arrives. Moreover, given a time horizon $T+1$ and integral transit times, then no optimal, integral, temporally repeated flow uses a path with length longer than $T$.
	
	Now, we compare $\val(f)$ with the number of trains that are sent in the convoy routing:
	\begin{align*}
        \val(f)
        &\leq \val(f^*)
        = \sum_{i \in [k]} (T +1 - \tau(P_i))
        = \Delta \cdot  \sum_{i \in [k]} \left(\frac{T +1 - \tau(P_i)}{\Delta} \right)\\
        &\leq \Delta \cdot  \sum_{i \in [k]} \left( 1+ \left\lfloor \frac{T - \tau(P_i)}{\Delta} \right\rfloor \right) = \Delta \cdot \sum_{i \in [k]} \sigma_i,
    \end{align*}
    Here the last inequality follows from integrality of $T - \tau(P_i)$ and $\Delta$.
    Thus, the convoy routing sends at least $\frac{1}{\Delta} \cdot \val(f)$ trains.
 (We can remove some trains to achieve exactly $\frac{1}{\Delta} \cdot \val(f)$ if we have too many in our computed routing.) 
\end{proof}

	Using \Cref{lem:routing-flow} and \Cref{lem:flow-routing}, we get an additive $\Delta$ approximation:

\begin{proof}[Proof of \Cref{thm:FoT-APX}]
	First, we solve the quickest flow over time problem on $D$ with unit capacities and demand $\Delta \cdot d$ in strongly polynomial time~\cite{saho2017cancel}.
	By \Cref{lem:routing-flow}, this returns a time horizon $T \leq \OPT + \Delta$.
	Then applying \Cref{lem:flow-routing} yields a convoy routing with at least $d = \frac{1}{\Delta} \cdot (\Delta \cdot d)$ trains and time horizon (at most) $T\leq \OPT + \Delta$.
\end{proof}

\section{Reduction to MinMaxDP – Proof of \Cref{thm:BlackBoxReduction}}
\label{app:Reduction}
In this section, we formally prove \Cref{thm:BlackBoxReduction} from \Cref{sec:Reduction}.
\BlackBoxReduction*

We start showing that (for the correct choice of $k$), we can transform the instance of \ac{TMO} to an instance of \ac{MinMaxDP}.
Therefore, for $k \in [m]$, let $D'_{k}$ be the graph obtained by adding the gadget $G_k$ from \Cref{fig:prefix_gadget} prior to the graph $D$ of our \ac{TMO} instance, \ie the source of $G_k$ is the source $s'$ of $D_k'$ and the sink of $G_k$ gets contracted with the source $s$ of $D$.

\begin{lemma}
    \label{lem:BlackBoxGadget}
    Let $k \in [m]$. A convoy routing in $D$ using $k$ disjoint paths can be transformed into a solution to \ac{MinMaxDP} in $D_k'$ with the same objective value, and vice versa.
\end{lemma}

\begin{proof}   
    By construction, any feasible solution $\bP'=(P'_1, \dots, P'_k)$ to \ac{MinMaxDP} in $D_k'$ uses all the $d-k$ arcs in $G_k$ that have length $\Delta$. This allows us to construct a feasible convoy routing $(\bP=(P_1, \ldots, P_k), \bsigma)$ to
    \ac{TMO} in $D$ with the same objective value. 
We define $P_i$ as the restriction of $P_i'$ to $D$, and $\sigma_i=|\{ a \in (G_k\cap P_i) \mid \tau_a = \Delta\}|+1.$

    Conversely, any feasible convoy routing $(\bP=(P_1, \dots, P_k), \bsigma)$   to \ac{TMO} in $D$ can be extended to $k$ pairwise arc-disjoint paths in $D'$ via the following rule: initially, mark all arcs in the gadget $G_k$ depicted in Figure~\ref{fig:prefix_gadget}  as \emph{unused}. Iteratively, for $i=1$ to $k$, select $\sigma_i-1$ \emph{unused} arcs with travel time $\Delta$ from $G_k$ and extend them using \emph{unused} arcs with zero travel time to an $s'$-$s$-path $Q_i$ (then mark all the used arcs as \emph{used}).
    Since $\sum_{i \in [k]} \sigma_i =d$, this construction is feasible. Moreover, by construction, the paths $P_i':=Q_i\cup P_i,\ i\in [k]$, form pairwise disjoint paths in $D_k'$, whose maximum travel time coincides with the makespan of the convoy routing $(\bP, \bsigma)$.
\end{proof}

Now, we are ready to prove \Cref{thm:BlackBoxReduction}
\begin{proof}[Proof of Theorem~\ref{thm:BlackBoxReduction}]
    First, note that the number of disjoint paths used in any solution is bounded by the number of arcs $m$. Moreover, let $\OPT$ be the optimal objective value of the \ac{TMO}-instance $(D,\tau,d,\Delta)$.
    We distinguish two cases:
    
    If $d > m (1+ \frac{1}{\epsilon})$, any optimal convoy routing -- which exists by Theorem~\ref{thm:ConvoyConjecture} –– features a path with more than $\frac{1}{\epsilon}+1$ trains assigned to it.
    Thus the last train of these paths arrives at least $\frac{1}{\epsilon} \cdot \Delta$ time units after the first one.
    As a result $\frac{1}{\epsilon} \cdot \Delta \leq \OPT$, or alternatively, $\Delta \leq \epsilon \cdot \OPT$. Hence, the flow-based algorithm from Theorem~\ref{thm:FoT-APX} yields a $(1+\epsilon)$-approximation in this case.

    Conversely, if $d \leq m(1+\frac{1}{\epsilon})$, we aim to use the connection to \ac{MinMaxDP} given by \Cref{lem:BlackBoxGadget}. We will guess the number of disjoint paths $k$ used in an optimal solution of $\ac{TMO}$. (We can do so by iterating $k$ from 1 to $m$.) 
    Given $k$, we apply the $\alpha$-approximation algorithm for \ac{MinMaxDP} to $D_k'$ (the series combination of the gadget of \Cref{fig:prefix_gadget} with graph $D$). By \Cref{lem:BlackBoxGadget} the optimal value of the \ac{MinMaxDP} instance corresponds to the optimal value of \ac{TMO}, \ie $\OPT$.
    By applying our $\alpha$-approximation algorithm to $D_k'$, we obtain a solution to \ac{MinMaxDP} of value at most $\alpha \cdot \OPT$. 
    We can transform this solution into a convoy routing for \ac{TMO} with the same value, \ie $\alpha \cdot \OPT$.

    To see the runtime, recall that the gadget $G_k$ added in front of $D$ has $(d-k)\cdot k$ arcs. Since $d \leq m(1+\frac{1}{\epsilon})$ by case distinction, we can bound the size by
    \begin{equation*}
        (d-k)\cdot k +m \leq m(1+\frac{1}{\epsilon})\cdot m +m \leq m^2(1+\frac{1}{\epsilon})+m \leq m^2(2+\frac{1}{\epsilon})
    \end{equation*}
    arcs for the constructed digraph $D_k'$. As we solve such an instance of \ac{MinMaxDP} for every $k \in [m]$, the result follows.
\end{proof}

%------------------------------------------------------------------------------------
 
\section[Proof that the parameter $\varphi(D)$ is well-defined.]{Proof that the parameter $\boldsymbol{\varphi(D)}$ is well-defined.}
%TODO Lennart
\label{app:phi_unique}

Let $D$ be a \ac{SePa} graph. In this section, we show the following lemma.

\begin{lemma}
    The parameter $\phi(D)$ is well-defined, \ie it does not depend on the choice of the binary decomposition tree of $D$.
\end{lemma}

To show that the parameter $\phi(D)$ is unique, we consider what we call the \emph{contracted decomposition tree} of $D$.
Given a binary decomposition tree $T_b$ of $D$, the contracted decomposition tree $T_c$ is obtained by merging adjacent vertices labeled with the same composition (series or parallel).
For example, the contracted decomposition tree for \Cref{fig:sepa-decomposition-tree} is shown in~\Cref{fig:contracted-sepa-decomposition-tree}.
The concept of contracted decomposition trees is not new, \eg in \cite{Bodlaender} it was called a ``minimal sp-tree''. Also the fact that the contracted decomposition tree is unique also seems to be a well-known statement (this is stated but not proven in \cite{Bodlaender}). For the sake of completeness, we give a proof of \Cref{lem:ContrComposTree} here.

\begin{figure}[thbp]
    \centering
    \begin{tikzpicture}[scale=0.8, transform shape, every node/.style={draw,circle,inner sep=0pt,minimum size=15pt, font=\scriptsize{#1}}]
        \node (1) at (0,2) {$a_1$};
        \node (2) at (1,2) {$a_2$};
        \node (3) at (2,2) {$a_3$};
        \node[rwthgreen] (p1) at (1,3) {$P$};

        \node (4) at (2.5,1) {$a_4$};
        \node[rwthgreen] (p2) at (3.5,1) {$P$};
        \node (5) at (3,0) {$a_5$};
        \node (6) at (4,0) {$a_6$};
        \node[rwthblue] (s1) at (3,2) {$S$};
        \node (7) at (4,2) {$a_7$};
        \node[rwthgreen] (p3) at (3.5,3) {$P$};

        \node (8) at (5,2) {$a_8$};
        \node (9) at (6,2) {$a_9$};
        \node (10) at (7,2) {$a_{10}$};
        \node[rwthgreen] (p4) at (6,3) {$P$};
        
        \node (12) at (8.5,2) {$a_{12}$};
        \node (13) at (9.5,2) {$a_{13}$};
        \node[rwthgreen] (p5) at (9,3) {$P$};
        
        \node (11) at (7.5,3) {$a_{11}$};
        \node[rwthblue] (s2) at (4.5,4.5) {$S$};
        \node (14) at (9,4.5) {$a_{14}$};
        \node[rwthgreen] (p6) at (7,5.5) {$P$};
        
        \draw (1) -- (p1);
        \draw (2) -- (p1);
        \draw (3) -- (p1);
        
        \draw (5) -- (p2);
        \draw (6) -- (p2);
        \draw (4) -- (s1);
        \draw (p2) -- (s1);
        \draw (s1) -- (p3);
        \draw (7) -- (p3);

        \draw (8) -- (p4);
        \draw (9) -- (p4);
        \draw (10) -- (p4);

        \draw (12) -- (p5);
        \draw (13) -- (p5);

        \draw (p1) -- (s2);
        \draw (p3) -- (s2);
        \draw (p4) -- (s2);
        \draw (p5) -- (s2);
        
        \draw (11) -- (s2);
        \draw (p4) -- (s2);
        \draw (14) -- (p6);
        \draw (s2) -- (p6);
    \end{tikzpicture}
    \caption{The contracted decomposition tree of the graph $D$ given in \Cref{fig:sepa-decomposition-tree}. Note that this is not a binary tree anymore.}
    \label{fig:contracted-sepa-decomposition-tree}
\end{figure}
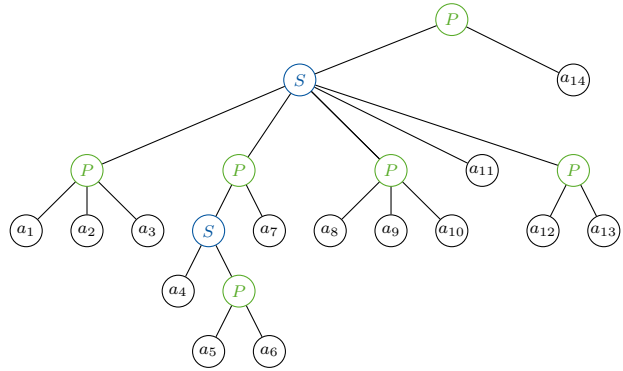

Note that $\phi(D)$ is equal to the maximal number of \encircle{$S$}-labeled vertices on a path from the root to a leaf in the contracted decomposition tree. Hence, the uniqueness of $\phi$ follows by the uniqueness of the contracted decomposition tree.

\begin{lemma}
    \label{lem:ContrComposTree}
    The contracted decomposition tree for any \ac{SePa} graph is unique.
\end{lemma}
\begin{proof}[Proof of \Cref{lem:ContrComposTree}]
    Let $D=(V,A)$ be a given \ac{SePa} graph. Note first that in a contracted decomposition tree, any path from the root to a leaf always alternates between \encircle{$S$} and \encircle{$P$} nodes.

    Assume for contradiction that there are two different contracted decomposition trees $T_c$ and $T_c'$ for $D$.
    Now, consider a vertex $v$ in $T_c$ with maximal distance to the root, such that the subtree with root $v$ in $T_c$ is not contained in $T_c'$. Since every leaf in both $T_c$ and $T_c'$ corresponds to an arc in $D$, $v$ must be an internal node.
    By definition, the children $w_1,\dots,w_\ell$ of $v$ are roots of subtrees in $T_c$ and $T_c'$ (\ie every subtree rooted in $w_i$ for $i \in [\ell]$ is the same in both trees by choice of $v$).
    Assume that the nodes $w_i$ for $i \in [\ell]$ do not have the same parent in $T_c'$. This is without loss of generality, since  if all $w_i$ for $i \in [\ell]$ have the same parent  $v'$ in $T_c'$, then $v'$ must have at least one more child, as otherwise we get a contradiction to the choice of $v$. In this case, we can swap the roles of $T_c$ and $T_c'$. 
    
    %(otherwise swap the roles of $T_c$ and $T_c'$ from hereon). 
    For $i \in [\ell]$, let $s_i,t_i \in V$ be the start and target nodes of the subgraphs corresponding to $w_i$.
    Further, let $A_i \coloneqq \{a \in A \mid a \text{ is contained in the subgraph corresponding to } w_i\}$.
    
    \begin{itemize}
        \item Assume that the label of $v$ is \encircle{$S$}.

        \Wlog the children $w_1,\dots, w_\ell$ of $v$ are ordered such that $t_i=s_{i+1}$ for $i \in [\ell-1]$.
        By construction with respect to $T_c$, for $i\in \{2,\dots\ell\}$ there is no arc in $A \setminus A_i$ leaving $s_i$ and for $i \in \{1,\dots, \ell-1\}$ no arc in $A \setminus A_i$ entering $t_i$ in the original graph $D=(V,A)$.
        
        Let $v_i'$ be the parent vertex of $w_i$ in $T_c'$.
        If there exists $j \in [\ell]$ such that vertex $v_j'$ is labeled with \encircle{$P$}, this is a contradiction:
        For $j \neq \ell$ there would be an arc in $A \setminus A_j$ that enters $t_j$. For $j=\ell$, there would be an arc in $A \setminus A_j$ leaving $s_\ell$.

        Hence, we can assume that $v_i'$ is labeled with \encircle{$S$} for $i \in [\ell]$.
        By definition of $v$, we know that there exists $j \in [\ell-1]$ such that $v_j' \neq v_{j+1}'$.
        Since $t_j =s_{j+1}$, and $w_j$, $w_{j+1}$ have different parents in $T_c'$, we have that $t_j$ must be the sink of the subgraph corresponding to $v_j'$ (and $s_{j+1}$ must be the source of the subgraph corresponding to $v_{j+1}'$).
        Since the contracted decomposition tree alternates between \encircle{$S$} and \encircle{$P$} compositions, the parent vertex of $v_{j}'$ (respectively $v_{j+1}'$) is \encircle{$P$}-labeled, and hence, there is an arc in $A \setminus A_j$ entering $t_j$ (and an arc in $A \setminus A_{j+1}$ leaving $s_{j+1}$).
        This is a contradiction.
        
        \item Assume that the label of $v$ is \encircle{$P$}.
        
        By construction with respect to $T_c$, $s_1=s_i$ and $t_1=t_i$ for $i \in [\ell]$.
        Let $v_i'$ be the parent vertex of $w_i$ in $T_c'$.
        If there exists an $i \in [\ell]$ such that vertex $v_i'$ is labeled with \encircle{$S$}, and since not all nodes $w_j$ have the same parent, this would be a contradiction since either $s_i$ or $t_i$ cannot be used for further compositions. But all subgraphs for the children $w_j$, $j \in [\ell]$ must be in parallel (and in particular, must have the same source and destination).
        Hence, we can assume that $v_i'$ is labeled with \encircle{$P$} for $i \in [\ell]$.
        
        By definition of $v$, we know that there exists a $j \in [\ell-1]$ such that $v_j' \neq v_{j+1}'$.
        Since we only consider \encircle{$S$}-\encircle{$P$}-alternating paths in the tree, the parent vertex of $v_{j}'$ (respectively $v_{j+1}'$) is \encircle{$S$}-labeled.
        Thus, either $s_j$ or $t_j$ cannot be used for further compositions, a contradiction since we still need to compose all the subgraphs for the children $w_j$, $j \in [\ell]$ in parallel. \qedhere
    \end{itemize}
\end{proof}

\section{Greedy algorithm for MinMaxDP – missing proofs}
\label{app:Greedy}
In this section, we give the missing proofs of \Cref{sec:Greedy}.

\subsection{Proof of \Cref{lem:H_k-combination}}
\label{app:lem:H_k-combination}
\HkCombination*

\begin{proof}
    Let $q_1 \geq \dots \geq q_{k_1}$ and $r_1 \geq \dots \geq r_{k_2}$ be the lengths of the paths in path profiles $\bQ$  and $\bR$, respectively. 
    In the following, we refer to these as  strict orders, \ie\ $q_i \prec q_j$ and $r_i \prec r_j$ if and only if $i<j$.

    Note that consistency of $\bP$ in $D'$ follows immediately from the consistency of $\bQ$ and $\bR$ in $D^{(1)}$ and $D^{(2)}$, respectively, since
    \begin{equation*}
        \cost(\bP, D') = \sum_{j=1}^{k_1} q_j + \sum_{j=1}^{k_2} r_j
        %= \cost(\bQ, D^{(1)})+\cost(\bR, D^{(2)})
        \overset{\eqref{eq:inv1}}{\leq} \cost(\bP^*, D^{(1)})+\cost(\bP^*, D^{(2)}) =
        \cost(\bP^*, D').
    \end{equation*}

    Now, we prove that the series composition $\bP = \bQ \serc \bR$ (thus, $k_1=k_2=:k'$) is balanced.
    Let $p_1 \geq \dots \geq p_{k'}$ be the lengths of the paths in profile $\bP$.
    We define $c \colon [k'] \to [k']$ to be the bijective function where $c(i)$ is the index such that $P_i$ is obtained by combining $Q_{c(i)}$ and $R_{k-c(i)+1}$.
    Without loss of generality, if $p_i=p_j$ with $i<j$, we assume that $c(i) < c(j)$. This way, we obtain once again a strict order on $\bP$.
    Using this ordering of the lengths of $\bP$ also gives us a strict order, \ie\ $p_i \prec p_j$ if and only if $i < j$.

    In order to show that $\bP$ is balanced in $D'$, we fix some $i \in [k'-1]$, and
     define
    \begin{equation*}
        S_1 \coloneqq \{1 \leq j \leq c(i+1)-1 \mid p_{c^{-1}(j)} \succ p_{i+1} \} \, \text{ and } \, S_2 \coloneqq \{c(i+1)+1 \leq j \leq k' \mid 
        p_{c^{-1}(j)} \succ p_{i+1} \}.
    \end{equation*}
    Note that $i=|S_1| + |S_2|$ by definition of the sets $S_1$ and $S_2$.
    Further, we define
    \begin{equation*}
        T_1 \coloneqq \{1 \leq j \leq c(i+1)-1 \mid p_{c^{-1}(j)} \prec p_{i+1}\}.
    \end{equation*}
    Then $[c(i+1)-1] = S_1 \dot\cup\ T_1$ holds by definition of $c$ and since $\prec$ is a strict ordering.
    Let $j_1 \leq \dots \leq j_{\ell}$ be the elements of $T_1$.
    Further, define $j_0 \coloneqq 0$ and $j_{\ell+1} \coloneqq c(i+1)$.
    Then $S_1 = \dot\bigcup_{h=0}^\ell I_h$ where $I_h \coloneqq [j_{h}+1,j_{h+1}-1]_{\mathbb{Z}}$ is the interval between $j_{h}$ and $j_{h+1}$ for $h \in \{0,\dots,\ell\}$.
    We obtain
    \begin{align}
    \label{eq:greedy_series}
    \begin{split}
        &\frac{1}{|I_h|} \sum_{j \in I_h} (q_j + r_{k'-j+1})
        \overset{\small (*)}{\leq} r_{k'-j_{h+1}+1} + \frac{1}{|I_h|} \sum_{j \in I_h} q_j
        \overset{\small (**)}{\leq} r_{k'-j_{h+1}+1} + \frac{1}{j_{h+1}-1} \sum_{j=1}^{j_{h+1}-1} q_j\\
        &\overset{\small \eqref{eq:inv2}}{\leq} r_{k'-j_{h+1}+1} + \makespan(\bP^*,D^{(1)}) + q_{j_{h+1}}
        \overset{\small (***)}{\leq} \makespan(\bP^*,D')  + p_{i+1}
    \end{split}
    \end{align}
    where $(*)$ holds since $r_{k'-j+1} \leq r_{k'-j_{h+1}+1}$ for $j \in I_h$.
    Equation $(**)$ holds since $q_1 \geq \dots \geq q_{j_h} \geq q_{j_h+1} \geq \dots \geq q_{j_{h+1}-1}$, so the average over all path lengths is larger than the average over the length of the $|J_h|$ shortest paths.
    The last equation $(***)$ holds since $ \makespan(\bP^*,D^{(1)}) \leq  \makespan(\bP^*,D')$ and since by definition $j_{h+1} \leq c(i+1)$ and $j_{h+1} \notin S_1$, so $ q_{j_{h+1}} + r_{k'-j_{h+1}+1} \leq p_{i+1}$.

    Since $S_1 = \dot\bigcup_{h=0}^\ell I_h$, it follows that
    \begin{equation}
        \label{eq:DP-corollary}
        \frac{1}{|S_1|} \sum_{j \in S_1} (q_j + r_{k'-j+1}) \leq \max_{h \in \{0,\dots,\ell\}} \frac{1}{|I_h|} \sum_{j \in I_h} (q_j + r_{k'-j+1}) \overset{\eqref{eq:greedy_series}}{\leq} \makespan(\bP^*,D') + p_{i+1}.
    \end{equation}

    By a similar argumentation (swapped roles for $q$ and $r$), we obtain
    \begin{equation}
        \label{eq:DP-corollary-b}
        \frac{1}{|S_2|} \sum_{j \in S_2} (q_j + r_{k'-j+1}) \leq  \makespan(\bP^*,D') + p_{i+1}.
    \end{equation}
    In summary, it follows that
    \begin{align*}
        \sum_{j=1}^i p_i 
        &= \sum_{j \in S_1} (q_j + r_{k'-j+1}) + \sum_{j \in S_2} (q_j + r_{k'-j+1})\\
        &\leq (|S_1| + |S_2|) \cdot (\makespan(\bP^*,D') + p_{i+1})
        = i \cdot ( \makespan(\bP^*,D') + p_{i+1})
    \end{align*}
    Thus, the greedy profile $\bP = \bQ \serc \bR$ is balanced in $D' = D^{(1)} \serc D^{(2)}.$
    \medskip

    Next, we consider the path profile  $\bP = \bQ \parc \bR$ obtained by a parallel composition of $\bQ$ and $\bR$ by the greedy algorithm. Let $p_1 \geq \dots \geq p_{k'}$ be the lengths of the paths in $\bP$ (thus, 
    $k'=k_1 + k_2$).
    To show that $\bP$ is balanced,  we fix some $i \in [k']$. Let $i_1$ and $i_2$ denote the numbers of paths from $P_1,\dots,P_i$ that traverse $D^{(1)}$ and $D^{(2)}$, respectively.
    If $i_1=k_1$, we get
    \begin{align*}
        \sum_{j=1}^{i_1} q_j &= \sum_{j=1}^{k_1} q_j = \cost(\bQ,D^{(1)}) \overset{\eqref{eq:inv1}}{\leq} \cost(\bP^*,D^{(1)}) \leq k_1 \cdot \makespan(\bP^*,D^{(1)})\\
        &\leq k_1 \cdot \makespan(\bP^*,D') \leq i_1 \left( \makespan(\bP^*,D') + p_{i+1} \right).
    \end{align*}
    If $i_1 < k_1$, we get
    \[
        \sum_{j=1}^{i_1} q_j \overset{\eqref{eq:inv2}}{\leq} i_1 \left( \makespan(\bP^*,D^{(1)}) + q_{i_1+1}\right) \leq i_1 \left( \makespan(\bP^*,D')  + p_{i+1} \right)
    \]
    where the last equality holds since $p_{i+1} = \max \{q_{i_1+1} , r_{i_2+1} \} \geq q_{i_1+1} $. Similarly, we obtain
    \[
        \sum_{j=1}^{i_2} r_j \leq i_2 \left( \makespan(\bP^*,D')  + p_{i+1} \right).
    \]
    Combining these equations, we conclude that
    \begin{align*}
        \sum_{j=1}^i p_i &= \sum_{j=1}^{i_1} q_j + \sum_{j=1}^{i_2} r_j
        \leq  i_1 \left( \makespan(\bP^*,D')  + p_{i+1} \right) + i_2 \left( \makespan(\bP^*,D') + p_{i+1} \right)\\
        &= i \left(\makespan(\bP^*,D') + p_{i+1} \right).
    \end{align*}
    Hence, $\bP = \bQ \parc \bR$ is balanced with $\bP^*$ in $D'=D^{(1)} \parc D^{(2)}$.
\end{proof}

\subsection{Proof of \Cref{lem:H_k-approx}}
\label{app:lem:H_k-approx}
\HkApprox*
\begin{proof}
    In this proof, we will use the following well known recursive formula for $H_i$:
    \begin{equation}
        \label{eq:H_i_recursion}
        \frac{1}{i} \left( \sum_{j=1}^{i-1} H_j \right) +1 = H_{i} \qquad \text{ for all } i \in \mathbb{Z}_{\geq 2}.
    \end{equation}
    Let $p_1 \geq \dots \geq p_k$ the lengths of the paths in $\bP$.
    Using the recursive formula \eqref{eq:H_i_recursion} and equation \eqref{eq:inv2} for $i \in [k-1]$, we start showing by induction that 
    \begin{equation}
        \label{eq:DP-Hk_corollary}
        p_1 \leq H_i \cdot \makespan(\bP^*,D) + p_{i+1} \qquad \text{ for all } i \in [k-1].
    \end{equation}
    For $i=1$ it follows directly from \eqref{eq:inv2}. 
    For the induction step, fix an $i \in [k-2]$ and assume that \eqref{eq:DP-Hk_corollary} holds for $i'\leq i$. We will now show that \eqref{eq:DP-Hk_corollary} also holds for $i+1$.
    \begin{align*}
        p_1 &\overset{\small (*)}{\leq} \frac{1}{i+1} \left(p_1 + \sum_{j=1}^i \left(H_j \cdot \makespan(\bP^*,D) + p_{j+1}\right)\right)\\
        &= \frac{1}{i+1} \cdot \makespan(\bP^*,D)  \left(\sum_{j=1}^i H_j \right) + \frac{1}{i+1} \left( \sum_{j=1}^{i+1} p_j \right)\\
        &\overset{\small (**)}{\leq} \frac{1}{i+1} \cdot \makespan(\bP^*,D) \left(\sum_{j=1}^i H_j \right) + \left( \makespan(\bP^*,D) + p_{i+2}\right)\\
        &\overset{\small \eqref{eq:H_i_recursion}}{=} H_{i+1} \cdot \makespan(\bP^*,D)  + p_{i+2}.
    \end{align*}
    Here, equation $(*)$ holds by averaging the equation $p_1 \leq p_1$ and equations \eqref{eq:DP-Hk_corollary} for $i'\leq i$ given by induction hypothesis. Equation $(**)$ holds by equation \eqref{eq:inv2} since $i+1\leq k-1$.

    Using equations \eqref{eq:DP-Hk_corollary} for $i \in [k-1]$, we obtain
    \begin{align*}
        \makespan(\bP,D)
        &= p_1 \overset{\eqref{eq:DP-Hk_corollary}}{\leq} \frac{1}{k} \cdot \left( p_1 + \sum_{i=1}^{k-1}(H_i \cdot \makespan(\bP^*,D) + p_{i+1}) \right)\\
        &= \frac{1}{k}  \cdot \makespan(\bP^*,D) \cdot \left( \sum_{i=1}^{k-1} H_i \right) + \frac{1}{k} \left( \sum_{i=1}^k p_i \right)\\
        &\overset{\small \eqref{eq:inv1}}{\leq} \frac{1}{k} \makespan(\bP^*,D) \cdot \left( \sum_{i=1}^{k-1} H_i \right) + \frac{1}{k} \cost(\bP^*,D)\\
        &\overset{\text{avg.}}{\leq} \frac{1}{k}  \cdot \makespan(\bP^*,D) \cdot \left( \sum_{i=1}^{k-1} H_i \right) + \makespan(\bP^*,D)\\
        &= \left( \frac{1}{k} \left(\sum_{i=1}^{k-1} H_i\right) +1 \right)
        \cdot \makespan(\bP^*,D)
        \overset{\small \eqref{eq:H_i_recursion}}{=} H_k \cdot \makespan(\bP^*,D).
        \qedhere
    \end{align*}
\end{proof}

\subsection{Proof of \Cref{lem:phi-combination}}
\label{app:lem:phi-combination}
\GreedyPhi*
\begin{proof}
    Let $p_1^{(j)} \geq \dots \geq p_{k_j}^{(j)}$ be the path lengths of path profile $\bP^{(j)}$ with $k_j$ paths for $j \in [i]$.
    
    Note that consistency of $\bP$ follows immediately from the consistency of $\bP^{(j)}$ for $j \in [i]$, since
    \begin{equation}
        \label{eq:inv1_proof}
        \cost(\bP,D') = \sum_{j=1}^i \tau\left(P^{(j)},D^{(j)}\right) \overset{\eqref{eq:inv1}}{\leq} \sum_{j=1}^i \cost\left(\bP^*,D^{(j)}\right) = \cost(\bP^*,D')
    \end{equation}
    We first consider the case where $D'$ is the root of an $S$-component, and show that  $D'$  is $\phi$-bounded.
    Note that $D'$ is the series composition of $D^{(1)}, \dots, D^{(i)}$, where each $D^{(j)}$ is a graph containing exactly one arc or is obtained by a parallel composition.
    Thus, $k=k_j$ for $j \in [i]$ and
    \begin{equation}
        \label{eq:phi_constr_cut}
        \phi(D') = \max_{j \in [i]} \left\{ \phi \left(D^{(j)} \right) \right\} +1.
    \end{equation}
    Let $p_1\geq \dots \geq p_k$ be the lengths of the paths in the greedy series composition of $\bP^{(1)}, \dots ,  \bP^{(i)}$ according to the fixed decomposition of $D'$. 
    Then, we have
    \begin{equation}
        \label{eq:phi_composition1}
        \begin{split}
        p_1 - p_k 
        &\overset{\small (*)}{\leq} \max_{j \in [i]} \left\{ p_1^{(j)} - p_k^{(j)} \right\}
        \overset{\small (**)}{\leq} \max_{j \in [i]} \left\{ \left(\phi(D^{(j)}) +1 \right) \cdot \makespan(\bP^*,D) \right\}\\
        &\overset{\eqref{eq:phi_constr_cut}}{\leq} \phi(D') \cdot \makespan(\bP^*,D),
        \end{split}
    \end{equation}
    where $(*)$ holds by the greedy composition rule, and $(**)$ holds since each $\bP^{(j)}$ is $\phi$-bounded in $D^{(j)}$.
    Further, we know that
    \begin{equation}
        \label{eq:phi_composition2}
        p_k \leq \frac{1}{k} \cost(\bP,D') \overset{\eqref{eq:inv1_proof}}{\leq} \frac{1}{k} \cost(\bP^*,D') \leq \makespan(\bP^*,D).
    \end{equation}
    We insert \eqref{eq:phi_composition2} into \eqref{eq:phi_composition1} and we get
    \begin{equation*}
        p_1 \leq \phi(D') \cdot \makespan(\bP^*,D) + p_k \leq (\phi(D')+1) \cdot \makespan(\bP^*,D). 
    \end{equation*}

    Next, we consider a $P$-component whose root corresponds to $D'$.
    Then $D'$ is obtained by  combining $D^{(1)}, \dots, D^{(i)}$ via  parallel compositions, where each $D^{(j)}$ either contains exactly one arc or is obtained by a series composition.
    Thus $k_1 + \dots + k_i=k$ and
    \begin{equation}
        \label{eq:phi_constr_cut_p}
        \phi(D') =\max_{j \in [i]} \left\{\phi(D^{(j)})\right\}
    \end{equation}
    Again, let $p_1 \geq \dots \geq p_k$ be the path lengths of the path profile $\bP$ obtained by greedily combining the path profiles $\bP^{(1)},\dots, \bP^{(i)}$ according to the fixed decomposition of $D'$.
    Here, $\bP$ is $\phi$-bounded, since $\bP^{(j)}$ is $\phi$-bounded in $D^{(j)}$ for $j \in [i]$:
    \begin{equation*}
        p_1 \leq \max_{j \in [i]} \left\{ p_1^{(j)} \right\}
        \leq  \max_{j \in [i]} \left\{ \left(\phi\left(D^{(j)}\right)+1\right) \cdot \makespan(\bP^*,D) \right\}
        \overset{\eqref{eq:phi_constr_cut_p}}{=} (\phi(D)+1) \cdot \makespan(\bP^*,D).
        \qedhere
    \end{equation*}
\end{proof}

\subsection{Proof of \Cref{DPHk}}
\label{app:lem:DPHk}

\DPHk*
\begin{proof}
    Let $\bP^*$ be the optimal path profile for $D$.

    Consider the path profiles of table entries along $\bP^*$, \ie\ the entries $(D',k',\theta')$ where $k'$ equals the number of paths of $\bP^*$ going through $D'$ and $\theta'$ equals their total length, \ie $\theta' = \cost(\bP^*,D')$.
    We show that these path profiles fulfill equations \eqref{eq:inv1} and \eqref{eq:inv2}. 
    For the single arcs this is trivially the case. For larger graphs, it follows iteratively (\wrt the series-parallel composition of $D$) by \Cref{lem:H_k-combination} and the fact that the left hand side of \eqref{eq:inv2} is minimized by our selection rule.
    
    Hence, in the last step, we pick a path profile $\bP$ whose objective value is at most as large as the objective value of path profile $\bP'$ in table entry $(D,k,\cost(\bP^*,D))$.
    Since $\bP'$ is consistent and balanced in $D$, it is a $H_k$-approximation by \Cref{lem:H_k-approx}. Thus,
    \begin{equation*}
        DP = \makespan(\bP,D) \leq \makespan(\bP',D)  \leq H_k \cdot \makespan(\bP^*,D)  = H_k \cdot \OPT. \qedhere
    \end{equation*}
\end{proof}

\subsection{Proof of \Cref{DPphi}}
\label{app:lem:DPphi}
\DPphi*
\begin{proof}
    Consider a graph $D' \sqsubseteq D$ that corresponds to a leaf in the decomposition tree of $D$ or to the root of an $S$- or $P$-component.
    Let $k'$ be the number of paths in $\bP^*$ that go through $D'$  and $\theta'=\cost(\bP^*,D')$ their total length in $D'$.
    First, we will show by induction that the path profile $\bP'$ in table entry $(D',k',\theta')$ with path length $p_1' \geq \dots \geq p_{k'}'$ is consistent and $\phi$-bounded.
    
    If $D'$ consists of only a single arc, the statement follows directly.

    Next we focus on an $S$- or $P$-component.
    Consistency follows directly, since we only consider path profiles for table entries along $\bP^*$, \ie\ path profiles with total length $\cost(\bP^*,D')$ and the correct number of paths.
        
    For now, assume that $D'$ corresponds to the root vertex of an $S$-component.
    Let $D^{(1)},\dots,D^{(i)} \sqsubseteq D$ be the subgraphs corresponding to children of the selected $S$-component.
    That means, $D'$ is obtained by series compositions of the subgraphs $D^{(1)},\dots,D^{(i)}$.
    Note that it might well be that the decomposition tree has multiple branches, \eg if $D' = ( D^{(1)} \serc D^{(2)}) \serc (D^{(3)} \serc D^{(4)})$.
    
    By construction, it holds that
    \begin{equation}
        \label{eq:phi_construction}
        \phi(D') = \max_{j \in [i]} \left\{ \phi \left(D^{(j)}\right) \right\} +1 \geq \phi\left(D^{(j)}\right).
    \end{equation}
    
    For $j \in [i]$, let $\bP^{(j)}$ be the path profile in table entry $\left(D^{(j)},k',\cost(\bP^*,D^{(j)})\right)$ with paths of length $p^{(j)}_1 \geq \dots \geq p_{k'}^{(j)}$.
    By our induction hypothesis $\bP^{(j)}$ is consistent and $\phi$-bounded in $D^{(j)}$.
    
    To show that the chosen solution for the table entry $(D',k',\theta')$ is $\phi$-bounded and fulfills equation~\eqref{eq:inv1}, we show the following claim:
    \begin{claim}
        \label{claim:p1-pk}
        Let $D_q$ be obtained by a series combination of $D^{(j)}$ with $j \in S_q \subseteq [i]$ and $D_r$ be obtained by a series combination of $D^{(j)}$ with $j \in S_r \subseteq [i]$ where $S_q \cap S_r = \emptyset$.
        Moreover, let $D_q \serc D_r = \tilde{D} \sqsubseteq D$.
        Then, the solution picked for $(\tilde{D}, k', \cost(\bP^*,\tilde{D}))$ fulfills
        \begin{equation}
            \label{eq:p1-pk}
            \tilde{p}_1 - \tilde{p}_{k'} \leq \max_{j \in S_q \cup S_r} \left\{p_1^{(j)}- p_{k'}^{(j)} \right\}
        \end{equation}
    \end{claim}
    \begin{claimproof}
        We prove this claim by induction.
        Let $q_1 \geq \dots \geq q_{k'}$ be the path length of the path profile $\bQ$ in table entry $(D_q,k',\cost(\bP^*,D_q))$ and let $r_1 \geq \dots \geq r_{k'}$ be the path length of the path profile $\bR$ in table entry $(D_r,k',\cost(\bP^*,D_r))$.
        We assume that \eqref{eq:p1-pk} holds for $\bQ$ and $\bR$, hence
        \begin{equation}
            \label{eq:p1-pk_IH}
            q_1 - q_{k'} \leq \max_{j \in S_q} \left\{ p_1^{(j)} - p_{k'}^{(j)} \right\} \qquad \text{ and } \qquad r_1 - r_{k'} \leq \max_{j \in S_r} \left\{ p_1^{(j)} - p_{k'}^{(j)} \right\}.
        \end{equation}
        For the path profile $\hat{\bP}$ obtained by the series combination of $\bQ$ and $\bR$, it holds that
        \begin{equation}
            \label{eq:p1-pk_IS}
            \hat{p}_1 - \hat{p}_{k'} \leq \max\left\{q_1 - q_{k'}, r_1 - r_{k'} \right\} \overset{\eqref{eq:p1-pk_IH}}{\leq} \max_{j \in S_q \cup S_r} \left\{ p_1^{(j)} - p_{k'}^{(j)} \right\}.
        \end{equation}
        Since $\hat{\bP}$ was a candidate for table entry $(D_r,k',\cost(\bP^*,\tilde{D}))$, and we pick the candidate $\tilde{\bP}$ minimizing the distance between the longest and the shortest path according to \eqref{eq:DP_phi_choice_s}, we get $\tilde{p}_1 - \tilde{p}_{k'} \leq \hat{p}_1 - \hat{p}_{k'}$.
        Together with \eqref{eq:p1-pk_IS} this proves the claim.
    \end{claimproof}

    By \Cref{claim:p1-pk}, we know that the path $\bP'$ in table entry $(D',k',\theta')$ fulfills
    \begin{equation}
        \label{eq:p1-k1:equation_series}
        p_1' - p_{k'}'
        \leq \max_{j \in [i]} \left\{ p_1^{(j)} - p_{k'}^{(j)}\right\}
        \leq \max_{j \in [i]} \left\{ \left(\phi \left(D^{(j)} \right) +1\right) \cdot \OPT \right\} \overset{\eqref{eq:phi_construction}}{\leq} \phi(D') \cdot \OPT.
    \end{equation}
    Thus, we obtain
    \begin{align*}
        p_1' &\overset{\eqref{eq:p1-k1:equation_series}}{\leq} \phi(D') \cdot \OPT + p_{k'}'
        \overset{\text{avg.}}{\leq} \phi(D') \cdot \OPT + \frac{1}{k'}\cost(\bP',D')
        \overset{\small (*)}{\leq} \phi(D') \cdot \OPT + \frac{1}{k'}\cost(\bP^*,D')\\
        &\overset{\text{avg.}}{\leq} \phi(D') \cdot \OPT + \cost(\bP^*,D')
        \leq \phi(D') \cdot \OPT + \OPT
        = (\phi(D')+1) \cdot \OPT.
    \end{align*}
    where $(*)$ holds since we used the solutions with $\cost(\bP^{(j)},D^{(j)})=\cost(\bP^*,D^{(j)})$ so $\cost(\bP',D') = \cost(\bP^*,D')$.
    \medskip
    
    Next, we focus on a $P$-component.
    Let $D^{(1)},\dots,D^{(i)} \sqsubseteq D$ be the subgraphs corresponding to children of the selected $P$-component.
    That means, $D'$ is obtained by parallel compositions of the subgraphs $D^{(1)},\dots,D^{(i)}$.
    For $j \in [i]$, let $\bP^{(j)}$ be the path profile in table entry $\left(D^{(j)},k_j,\cost(\bP^*,D^{(j)})\right)$ where $k_j$ is the number of paths of $\bP^*$ going through $D^{(j)}$ and let $p^{(j)}_1 \geq \dots \geq p_{k_j}^{(j)}$ be the paths lengths of $\bP^{(j)}$.
     
    Again, it might be the case that the component has multiple branches, \eg if $D' = ( D^{(1)} \parc D^{(2)}) \parc (D^{(3)} \parc D^{(4)})$.
    By construction, it holds that
    \begin{equation}
        \label{eq:phi_construction_parallel}
        \phi(D') = \max_{j \in [i]} \left\{ \phi \left(D^{(j)}\right) \right\} \geq \phi\left(D^{(j)}\right).
    \end{equation}
    Since we choose the candidate path profile that minimizes the longest path according to \eqref{eq:DP_phi_choice_p} and as selecting $k_j$ paths in $D^{(j)}$ yields one of the candidate profiles considered, we obtain
    \begin{equation}
        \label{eq:phi_p1_parallel}
        p_1' \leq \max_{j \in [i]} \left\{ p_1^{(j)}\right\}.
    \end{equation}
    Since $\bP^{(j)}$ is $\phi$-bounded, we get
    \begin{equation}
        p_1'
        \overset{\eqref{eq:phi_p1_parallel}}{\leq} \max_{j \in [i]} \left\{ p_1^{(j)}\right\}
        \leq \max_{j \in [i]} \left\{ \left( \phi\left(D^{(j)} \right) +1 \right) \cdot \OPT \right\}
        \overset{\eqref{eq:phi_construction_parallel}}{\leq} \left( \phi(D') +1 \right) \cdot \OPT.
    \end{equation}
    
    \medskip
    Hence, we proved by induction on the decomposition tree (since $D'$ corresponds to a root of a $S$- or $P$-component) that the path profile $\bP'$ in table entry $(D,k,\cost(\bP^*,D))$ is consistent and $\phi$-bounded.
    In the end of the dynamic program, the path profile $\bP'$ is a potential candidate. Since we select the path profile with smallest makespan, we get
    \begin{equation*}
        DP \leq p_1' \leq (\phi(D) +1) \cdot \OPT. \qedhere
    \end{equation*}
\end{proof}

\subsection{Running time – Proof of \Cref{lem:DPruntime} and \Cref{lem:DPepsruntime}}
\label{app:DP_runtime}
\DPruntime*
\begin{proof}
    First note that we have to fill $\mathcal{O}(m k |\Theta|)$ table entries.
    For a table entry, we have to consider $k \cdot |\Theta|$ candidates.
    Constructing a path profile for a candidate takes $\mathcal{O}(k\log(k))$ time (sorting the paths by length).
    Evaluating \eqref{eq:DP_choice} for a candidate takes $\mathcal{O}(k^2)$ time.
    So it takes $\mathcal{O}(k^3 |\Theta|)$ time to fill a table entry.
    Thus, the total runtime is bounded by $\mathcal{O}(mk^3|\Theta|^2)$.
\end{proof}

A simple bound on $\Theta$ would be $m \cdot \tau_{\max}$.
However, this can be an extremely large number.
Especially, since it is possible that an artificial arc from $s$ to $t$ is added with arbitrary high travel time.
Even if this arc is never used in an optimal solution, the DP needs to consider this possibility and thus has pseudo-polynomial runtime.

To circumvent this, we guess the longest arc used by an optimal solution (at most $m$ possibilities) and round the arc length accordingly.
\DPepsruntime*

\begin{proof}
    To show that such an algorithm exists, we will describe how to use a rounding (depending on $\epsilon$) such that first $|\Theta|$, the number of different possibilities of the total length, can be bounded by $m$ and $\frac{1}{\epsilon}$. And second, a solution to the rounding adds a factor of at most $(1+\epsilon)$ to the $H_k$-approximation ratio of the dynamic program.

    For the rounding, we aim to guess the longest arc that is used in an optimal solution. We do this by going through all arcs $a \in A$.
    For a given arc $\bar a \in A$ with length $\tau_{\bar a}$, let $D'$ be the graph $D$ where all arcs $a\in A$ with $\tau_{a} > \tau_{\bar a}$ are deleted.
    Let $F \coloneqq \frac{\epsilon \cdot \tau_{\bar a}}{m}$ and define $\tau'_a \coloneqq \left\lceil \frac{\tau_a}{F}\right\rceil$. Note that this means instead of iterating over all arcs, we basically just iterate over the different length that arcs can have.
    
    Let $\bP(\tau_{\bar a})$ be the path profile obtained by applying the DP to $D'$ with arc length $\tau'$.
    Finally, we pick the path profile $\bP$ with minimal objective among all computed solutions $\bP(\tau_{\bar a})$, \ie
    $\bP = \text{argmin} \{\makespan(\bP(\tau_{a})) \mid a \in A \} $.
    \medskip

    The computed path profile $\bP$ has an objective value that is at most as long as the one of $\bP(\tau_{\max}^*)$, where $\tau_{\max}^*$ is the length of the longest arc in the optimal path profile $\bP^*$.

    To simplify notation, we use $F$ for the rounding divisor for $\tau_{\max}^*$, $\tau'$ for the rounded transit times for $\tau_{\max}^*$ and $\makespan'$ to map a path profile to its solution value with respect to the rounded transit times for $\tau_{\max}^*$.
    Similarly, let $\bP' \coloneqq \bP(\tau_{\max}^*)$ be the solution computed by the dynamic program with respect to this rounding according to $\tau_{\max}^*$.
    Further let $\bar{\bP}^*$ be an optimal path profile in $D$ with respect to travel times $\tau'$.
    Then, the following equations hold by definition of $\makespan'$:
    \begin{equation}
        \label{eq:DP-rounding_lb}
        0 \leq F \cdot\sum_{a \in P_1'} \left( \left\lceil \frac{\tau_a}{F} \right\rceil - \frac{\tau_a}{F} \right) = F \cdot \makespan'(\bP') - \makespan(\bP')
    \end{equation}
    here $P_1'$ expresses the longest path of $\bP'$ and 
    \begin{equation}
        \label{eq:DP-rounding_ub}
        F \cdot \makespan'(\bP^*) - \makespan(\bP^*) = F \cdot\sum_{a \in P_1^*} \left( \left\lceil \frac{\tau_a}{F} \right\rceil - \frac{\tau_a}{F} \right) \leq F \cdot \sum_{a \in P_1^*} 1 \leq mF.
    \end{equation}
    Using these two equations, we can show that $\tau(\bP) \leq H_k (1+ \epsilon) \tau(\bP^*)$:
    \begin{align*}
        \makespan(\bP)
        &\overset{\small (*)}{\leq} \makespan(\bP')
        \overset{\eqref{eq:DP-rounding_lb}}{\leq} F \cdot \makespan'(\bP')
        \overset{\small (**)}{\leq} F \cdot H_k \cdot \makespan'(\bar{\bP}^*)
        \overset{\small (***)}\leq F \cdot H_k \cdot \makespan'(\bP^*)\\
        &\overset{\eqref{eq:DP-rounding_ub}}\leq F \cdot H_k \cdot \left( m + \frac{\makespan(\bP^*)}{F} \right)
        = H_k \cdot (\epsilon \cdot \tau_{\max}^* + \makespan(\bP^*))\\
        &\overset{\small (****)}{\leq} H_k \cdot (1+ \epsilon) \cdot\makespan(\bP^*)
    \end{align*}
    Here, $(*)$ holds as the algorithm chooses a path profile minimizing the makespan ampong all possible roundings, $(**)$ holds by \Cref{DPHk}, $(***)$ holds by definition of $\bar{\bP}^*$ and $(****)$ holds since we chose $\tau_{\max}^*$ as the longest arc in the optimal solution.
    \medskip
    
    Lastly, we show the that the running time is polynomial in $m$ and $\frac{1}{\epsilon}$.
    By the rounding, we get that $\Theta \leq m\left(m \cdot \frac{1}{\epsilon}+1\right)$.
    Thus, rounding and running the dynamic program takes $\mathcal{O}\left(m^5k^4 \left(\frac{1}{\epsilon}\right)^2\right)$.
    There are at most $m$ different rounded instances and choosing the solution with the minimal longest path takes at most $\mathcal{O}(m\log(m))$ in addition.
    Hence, the total running time is bounded by $\mathcal{O}\left( m^6k^4\left(\frac{1}{\epsilon}\right)^2\right)$.
\end{proof}
\end{document}